%% file: main-TW-Jan-13-0169.tex
\documentclass[10pt,journal]{IEEEtran}

\usepackage{cite,url}
\usepackage{booktabs}
\usepackage{amsmath,amssymb,bm}
\usepackage[dvips]{color}
\usepackage{graphicx}
\usepackage[lined,boxed,commentsnumbered, ruled]{algorithm2e}
\usepackage{booktabs}
\usepackage{subfigure}
\usepackage{amsfonts}
\newtheorem{theorem}{Theorem}
\newtheorem{lemma}{Lemma}

\begin{document}
%


\title{Insensitivity of User Distribution in \\Multicell Networks under General Mobility \\and Session Patterns}

\author{Wei Bao,~\IEEEmembership{Student Member,~IEEE} and
        Ben Liang,~\IEEEmembership{Senior Member,~IEEE}
\thanks{A preliminary version of this work has appeared as \cite{infocom_version}.} }

\maketitle


\begin{abstract}
The location of active users is an important factor in the performance analysis of mobile multicell networks, but it is difficult to quantify due to the wide variety of user mobility and session patterns.  
In this work, we study the stationary distribution of users by modeling the system as a multi-route queueing network with Poisson inputs.  We consider arbitrary routing and arbitrary joint probability distributions for the channel holding times in each route.  Through a decomposition-composition approach, we derive a closed-form expression for the joint stationary distribution for the number of users in all cells. The stationary user distribution (1) is insensitive to the user movement patterns, (2) is insensitive to general and dependently distributed channel holding times, (3) depends only on the average arrival rate and average channel holding time at each cell, and (4) is completely characterized by an open network with $M/M/\infty$ queues.  We use the Dartmouth trace to validate our analysis, which shows that the analytical model is accurate when new session arrivals are Poisson and remains useful when non-Poisson session arrivals are also included in the data set. Our results suggest that accurate calculation of the user distribution, and other associated metrics such as the system workload, can be achieved with much lower complexity than previously expected.

\end{abstract}

\begin{IEEEkeywords}
Mobility modeling, multicell network, user distribution, insensitivity, dependent channel holding times.
\end{IEEEkeywords}

\input{section1-intro}

\input{section2-related}

\input{section3-model}

\input{section4-analysis}

\input{section5-experiment}

\input{section7-conclusion}

\end{document}

%% file: section1-intro.tex
\section {Introduction} \label{section_introduction}

\IEEEPARstart{I}n designing ever more efficient and capable mobile access networks, the accurate modeling of how user mobility and session connectivity patterns affect network performance is of paramount interest.  However, compared with wired networks, the analytical modeling of mobile networks is burdened with many additional technical challenges.  Some of the most difficult factors are the following:
\begin{itemize}
\item The movement of users may be individually arbitrary, without following any common mobility pattern \cite{Mobility-Survey}.
\item The session durations may have a general probability distribution, supporting diverse data and multimedia applications \cite{Mobility-Parameter}.
\item The channel holding times at different cells are correlated, dependent on the speed or trajectory of different users \cite{zahran:TMC08}.
\end{itemize}

To facilitate tractable analysis, existing studies often adopt simplified models.
For example, classical models assume exponentially distributed cell dwell times and session durations, resulting in independent and memoryless channel holding times  \cite{Guerin1987,Hong1986,Lin1997}. Although more general mobility and session models have been considered in the past literature \cite{FanChLi97,Mobility-Parameter,Wong2000,Fang2005,Marsan2004,zahran:TMC08}, to the best of our knowledge, none addresses all of the challenges above.

In this paper, we study the joint stationary distribution for the number of users in all cells in a multicell network, which has important utilization in network management and planning.  Prior studies have proposed several analytical models to estimate the user distribution with various degrees of detail and generality \cite{Mobility-RWP-Analytical,Mobility-Core2,Mobility-Core5,Mobility-Core4,Mobility-Core1}.  Instead, we consider general mobility and session patterns, only requiring that the new session arrivals form a Poisson process, which is well supported by experimental data \cite{Iversen:ITU01,Heegaard:LNCS07,Mobility-Core1}.  We model the user mobility with a general system with multiple routes, each representing one type of users with a specific movement pattern.  A general probability distribution is used to represent the session durations. As a consequence, the channel holding times at different cell sites are no longer independent.

Through a decomposition-composition approach, we derive a closed-form expression for the joint stationary distribution for the number of users in all cells.  We observe five important conclusions on the stationary user distribution: \emph{first}, it is insensitive to how the users move through the system; \emph{second}, it is insensitive to the general distribution of channel holding times; \emph{third}, it is insensitive to the correlation among the channel holding times; \emph{fourth}, it depends only on the average arrival rate and average channel holding time at each cell; and \emph{fifth}, it is perfectly captured by an open Jackson network with $M/M/\infty$ queues.

We confirm our theoretical analysis through experimental validation using the Dartmouth user mobility traces \cite{Trace-Dartmouth-Data03}.  These traces provide a large data set, with 152 APs and more than 5000 users, to support an accurate real-life measurement of the joint user distribution.  They also contain a large amount of handoff traffic among APs to create strong dependency between channel holding times, which tests our claim of insensitivity.  The experimental results show that the proposed analysis accurately predicts real-world user distributions.

The conclusion of this work has important consequence to performance analysis and practical system design.  It suggests that accurate calculation of the user distribution, and other associated metrics such as the system workload, can be achieved with much lower requirement for system parameter estimation than previously expected.  Furthermore, the simplicity of the resultant product-form user distribution enables further analytical endeavors in system optimization.


The rest of the paper is organized as follows. In Section \ref{section_relatedwork}, we discuss the relation between our work and  prior works. In Section \ref{section_assumption}, we present the system model. In Sections \ref{single} and \ref{multiple}, we derive the analytical stationary distributions for single-route and multiple-route networks, respectively. In Section \ref{section_experiment}, we validate our analysis with experimental results from the Dartmouth traces.  
Finally, concluding remarks are given in Section \ref{section_conclusion}.

%% file: section2-related.tex
\section {Related Work} \label{section_relatedwork}

In the following subsections, we present the related prior work in user distribution modeling and general results in the insensitivity of queueing networks.

\subsection{User Distribution in Mobile Networks}

The user distribution is an important factor in the management and planning of mobile networks.  However, relatively few analytical models are available in the literature.  The uniform user distribution has been widely adopted for mathematical convenience, but it does not account for non-homogeneity in the physical topology and is incorrect in some cases.  For example, it is well known that the user distribution is non-uniform under the random waypoint model \cite{Mobility-RWP-Analytical}.

Other previous works have proposed analytical models using stochastic queueing networks to derive the user distribution in different environments, including wireless multimedia networks \cite{Mobility-Core2}, vehicular ad-hoc networks \cite{Mobility-Core5}, and WLANs  \cite{Mobility-Core4,Mobility-Core1}.  However, they do not allow arbitrary mobility or arbitrary session patterns.  In terms of user movement, \cite{Mobility-Core2}, \cite{Mobility-Core5}, and \cite{Mobility-Core1} assume that users move from one cell to another probabilistically and memorylessly, while \cite{Mobility-Core4} focuses on scattered single cells, so that user movement among multiple cells is not discussed.  None of them consider the arbitrary user movement patterns. In terms of channel holding times, \cite{Mobility-Core2} uses the sum of hyper-exponentials or the Coxian distribution to approximate arbitrary distributions; \cite{Mobility-Core4} assumes generally distributed channel holding times but concerns only a single cell; and \cite{Mobility-Core5} and \cite{Mobility-Core1} consider generally but independently distributed channel holding times in different cells. None of the above works consider the dependence among channel holding times.

 Note that the authors of \cite{Mobility-Core1} have also observed a surprising match between analysis and real-life user mobility traces from the Dartmouth study \cite{Trace-Dartmouth-Data03}, even though their analysis assumes simple $M/G/\infty$ mobility and session models without considering arbitrary user movement patterns or dependent channel holding times.
 No analytical explanation is given in \cite{Mobility-Core1} for this observation.  In contrast, our work provides theoretical support for it, since we show that the stationary user distribution is also insensitive to arbitrary user movement patterns and dependent channel holding times.

\subsection{Insensitivity Property}

The insensitivity of queueing networks indicates the situation where the stationary distribution remains unchanged while the distribution of service times takes arbitrary forms.  When the service times are assumed independent among different queues, there are several well known conditions for insensitivity.  For example, networks with symmetric queues are insensitive \cite{Book-Reversibility}.  In \cite{Insensitivity1} and \cite{Insensitivity2}, the partial balance of probability flows is shown to be a sufficient condition for insensitivity.  In \cite{Insensitivity4}, partial reversibility is shown to be a necessary and sufficient condition for insensitivity.  However, none of these known results consider the case where the service times between different queues are dependent.  For example, the queueing network closely related to ours is one with $M/G/\infty$ queues.   It is known to be insensitive when the service times are independent \cite{Book-Reversibility}, but to the best of our knowledge, there is no further general result for dependent service times.

\subsection{Preliminary Version}
A preliminary version of this work has appeared as \cite{infocom_version}. This full version includes the following extensions: First, we provide more detailed derivation and discussion in the analysis of the single-route network in Sections \ref{single}.  Second, we fully expand the analysis of the multiple-route network by proving the theorem of insensitivity and deriving the stationary user distribution of multicell networks in Section \ref{multiple}. Third, we include new experimental studies in Section \ref{section_experiment}.

%% file: section3-model.tex
\section {System Model}\label{section_assumption}

Consider a cellular network with $C$ cells.  There are $L$ unique \emph{routes}, each defined as a finite ordered sequence of cells.  The $j$th stage on the $l$th route corresponds to the $j$th cell in the sequence, which is denoted as $c(l,j)$.  Let $N_l$ be the number of stages on the $l$th route.  Each user of the $l$th route
starts a new session in cell $c(l,1)$; then it moves along the route through cells $c(l,1), c(l,2)\ldots c(l,N_l)$, as long as the session remains active.  The user is considered to have departed the network when its session terminates or when it exits cell $c(l,N_l)$.  We allow an arbitrary number of arbitrary routes to cover all possible movement patterns.

For each route, we assume the arrivals of \emph{new} sessions to form a Poisson process.  Note that although the arrivals of packets in the Internet may not form Poisson processes \cite{Poisson-against1}, 
the arrivals of new sessions are at a much larger time scale and are well justified as Poisson \cite{Iversen:ITU01,Heegaard:LNCS07}.
Furthermore, in \cite{Mobility-Core1} and later in Section \ref{section_experiment}, experimental data show that new sessions in the type of mobile networks under consideration are indeed Poisson barring some extreme cases.  We emphasize that only the new session arrivals are Poisson, while the handoff arrivals at each cell have general statistics with complicated dependencies.

The session duration of a user on the $l$th route is modeled as an arbitrarily distributed random variable $T_{l}$.  Let $\lambda_{l0}$ be the new session arrival rate at the $l$th route.
After a new session arrival, let $\tau_{l1}$ denote the residual cell dwell time of the user in the $1$st stage on the $l$th route, which is arbitrarily distributed.  Let $\tau_{lj}$, $2 \leq j \leq N_l$, denote the cell dwell time of the user in the $j$th stage on the $l$th route, which are also arbitrarily distributed.  Then, the channel holding time of the $j$th stage on the $l$th route, $t_{lj}$, if it exists, can be represented as follows:
\begin{equation}
t_{lj}=
\begin{cases}
\min\{T_l, \tau_{l1}\}, & \textrm{ if } j=1,\\
\displaystyle
\min\{T_l-\sum_{i=1}^{j-1}\tau_{li}, \tau_{lj}\},& \textrm{ if } T_l>\sum_{i=1}^{j-1}\tau_{li}, 2\leq j\leq N_l.
\end{cases}
\end{equation}

Fig. \ref{systemmodel} shows an example network with $3$ routes.  Route $1$ starts from cell $1$ and passes cell $3$, $4$ and $6$ (i.e., $c(1, 1) = 1$, $c(1, 2) = 3$, $c(1, 3) = 4$ and $c(1, 4) = 6$). A user starts a session in cell $1$, and the session is terminated in cell $4$.  The corresponding $T_1, \tau_{11}, \tau_{12}, \tau_{13}, t_{11}, t_{12}$, and $t_{13}$ are labeled in the figure.

\begin{figure}[tbp]
\begin{center}
\scalebox{0.57}{ \input{systemmodel2.pstex_t} }
\caption{System model.}
\label{systemmodel}
\end{center}
\end{figure}
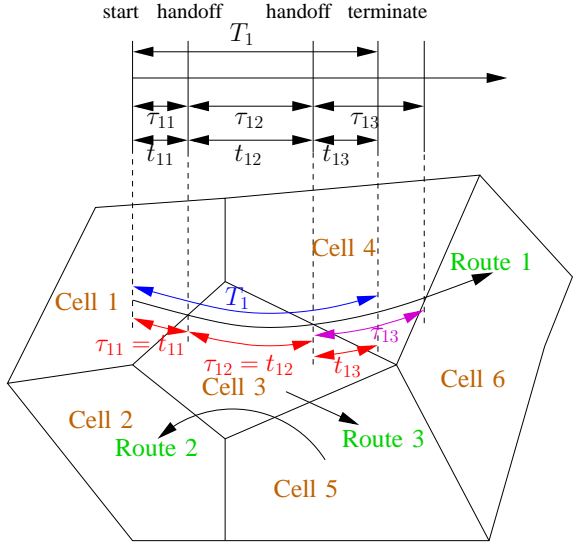

\begin{table}[ht]
\hspace{-0.4cm}
\renewcommand{\arraystretch}{1.3}
\caption{Selected Definition of Variables}
\label{table}
\small
\begin{tabular}{|c| c |   }
\hline
\bfseries Name  & \bfseries Definition \\
\hline
$x_{lj}$,                 & Number of sessions in the $j$th stage on the $l$th route,            \\
$\mathbf{x}$                 & $\mathbf{x}=[\{x_{lj}\}]^T$.            \\
\hline
$\widetilde{x}_{lkij}$        & On the $l$th route, number of sessions lasting $k$ stages, \\
$\widetilde{\mathbf{x}}$      &   in the $i$th  realization, in the $j$th stage,   $\widetilde{\mathbf{x}}=[\{\widetilde{x}_{lkij}\}]^T$. \\

\hline
$y_{n}$, $\mathbf{y}$       & The number of sessions in the $n$th cell,    $\mathbf{y}=[\{y_n\}]^T.$          \\
\hline
$t_{lj}$        & Random variable: on the $l$th route,  \\
                & channel holding time at the $j$th stage.  \\
\hline
$\widehat{t}_{lkj}$        & Random variable: on the $l$th route, channel holding  \\
                            &time at the $j$th stage  given that there are  $k$ stages.  \\
\hline
$\overline{t}_{lj}$        & On the $l$th route, the average value of $t_{lj}$, \\
                           &  given that the number of stages $\geq j$.  \\
\hline
$\mathbf{\widehat{t}_{lk}}$        & Random vector: $\{\widehat{t}_{lk1},\ldots,\widehat{t}_{lkk} \}$.\\
\hline
$\widetilde{t}_{lkij}$        & Constant: on the $l$th route, $i$th realization of channel \\
                              &  holding time at $j$th stage, when session lasts $k$ stages.\\
\hline
$\mathbf{\widetilde{t}_{lki}}$        & Constant vector:  $\{\widetilde{t}_{lki1},\ldots,\widetilde{t}_{lkik} \}$.\\
\hline

$p_{lk}$        &  On the $l$th route, probability \\
                &  that a session lasts $k$ stages.\\
\hline
$q_{lki}$        & On the $l$th route, probability of the $i$th realization of\\
                 & a session, given that a session lasts for $k$ stages.\\
\hline
$P_{lki}$        & On the $l$th route, probability that a session lasts $k$ \\
                 &  stages and is in the $i$th realization, $P_{lki}=p_{lk}q_{lki}$.\\
\hline
$\lambda_{l0}$        & Arrival  rate of the $l$th route.   \\
\hline
$\lambda_{lj}$        & $1/\overline{t}_{lj}$.   \\
\hline
$\widetilde{\lambda}_{lki0}$        &  Arrival  rate of sessions lasting $k$ stages, in the $i$th   \\
                                    &  realization, on the $l$th route, $\widetilde{\lambda}_{lki0}=P_{lki}\lambda_{l0}$. \\

\hline
$\widetilde{\lambda}_{lkij}$        &  $1/\widetilde{t}_{lkij}$.\\
\hline
$w'_{lj}$        &  Invariant measure of memoryless network. \\
\hline
$w_{lj}$        &  $w'_{lj}/\lambda_{lj}$. \\
\hline
$\widetilde{w}'_{lkij}$        &  Invariant measure of  decoupled  memoryless network. \\
\hline
$\widetilde{w}_{lkij}$        &  $\widetilde{w}'_{lkij}/\widetilde{\lambda}_{lkij}$. \\
\hline
$\pi_{0}(\mathbf{x})$        & Stationary distribution of memoryless network w.r.t $\mathbf{x}$. \\
\hline
$\pi_{D}(\widetilde{\mathbf{x}})$        & Stationary user distribution of decoupled network w.r.t $\mathbf{\widetilde{\mathbf{x}}}$.\\
\hline
$\pi(\mathbf{x})$        & Stationary user distribution of the original network w.r.t $\mathbf{x}$. \\
\hline
$\pi_1(\mathbf{y})$                 & Stationary user distribution of multicell network  w.r.t $\mathbf{y}$.     \\
\hline
$\overline{\lambda}_n$                 & Average arrival rate of the $n$th cell.    \\
\hline
$\overline{\mathfrak{t}}_n$                 & Average channel holding time of the $n$th cell.   \\
\hline
\end{tabular}
\end{table}

Since there is a one-to-one mapping between active users and sessions, we do not distinguish the two.
Note that given a route $l$, we know the sessions start at cell $c(l, 1)$ but the cell where they end is random.  Furthermore, we do not assume independence between $T_l$ and $\tau_{lj}$, and the channel holding times $t_{lj}$ are not independent either.
Finally, each route defines the user movement trace and the distribution of channel holding times, which implicitly characterizes the speed of users on this route.

Let $x_{lj}$, $1\leq l \leq L$, $1\leq j \leq N_l$, denote the number of active users in the $j$th stage on the $l$th route; let $y_n$, $1\leq n\leq C$, denote the number of active users in the $n$th cell. Let $\mathbf{x}=[\{x_{lj}:1\leq l \leq L, 1\leq j\leq N_l\}]^T$ and $\mathbf{y}=[y_1,y_2,\ldots,y_C]^T$. We aim to derive $\pi(\mathbf{x})$ and $\pi_1(\mathbf{y})$, the joint stationary user distributions for $\mathbf{x}$ and $\mathbf{y}$, respectively. Note that since $\pi(\mathbf{x})$ and $\pi_1(\mathbf{y})$ are defined in the steady state, we explicitly ignore any temporal fluctuation in these distributions.

A partial list of nomenclature is given in Table \ref{table}.

%% file: systemmodel2.pstex_t
\begin{picture}(0,0)%
\includegraphics{systemmodel2.pstex}%
\end{picture}%
\setlength{\unitlength}{3947sp}%
\begingroup\makeatletter\ifx\SetFigFont\undefined%
\gdef\SetFigFont#1#2#3#4#5{%
  \reset@font\fontsize{#1}{#2pt}%
  \fontfamily{#3}\fontseries{#4}\fontshape{#5}%
  \selectfont}%
\fi\endgroup%
\begin{picture}(6249,5988)(1939,-6064)
\put(3601,-286){\makebox(0,0)[lb]{\smash{{\SetFigFont{14}{16.8}{\rmdefault}{\mddefault}{\updefault}{\color[rgb]{0,0,0}handoff}%
}}}}
\put(4400,-541){\makebox(0,0)[lb]{\smash{{\SetFigFont{17}{20.4}{\rmdefault}{\mddefault}{\updefault}{\color[rgb]{0,0,0}$T_1$}%
}}}}
\put(4451,-1460){\makebox(0,0)[lb]{\smash{{\SetFigFont{17}{20.4}{\rmdefault}{\mddefault}{\updefault}{\color[rgb]{0,0,0}$\tau_{12}$}%
}}}}
\put(5727,-1460){\makebox(0,0)[lb]{\smash{{\SetFigFont{17}{20.4}{\rmdefault}{\mddefault}{\updefault}{\color[rgb]{0,0,0}$\tau_{13}$}%
}}}}
\put(3482,-1460){\makebox(0,0)[lb]{\smash{{\SetFigFont{17}{20.4}{\rmdefault}{\mddefault}{\updefault}{\color[rgb]{0,0,0}$\tau_{11}$}%
}}}}
\put(3482,-1868){\makebox(0,0)[lb]{\smash{{\SetFigFont{17}{20.4}{\rmdefault}{\mddefault}{\updefault}{\color[rgb]{0,0,0}$t_{11}$}%
}}}}
\put(4451,-1868){\makebox(0,0)[lb]{\smash{{\SetFigFont{17}{20.4}{\rmdefault}{\mddefault}{\updefault}{\color[rgb]{0,0,0}$t_{12}$}%
}}}}
\put(5421,-1868){\makebox(0,0)[lb]{\smash{{\SetFigFont{17}{20.4}{\rmdefault}{\mddefault}{\updefault}{\color[rgb]{0,0,0}$t_{13}$}%
}}}}
\put(4349,-3450){\makebox(0,0)[lb]{\smash{{\SetFigFont{17}{20.4}{\rmdefault}{\mddefault}{\updefault}{\color[rgb]{0,0,1}$T_1$}%
}}}}
\put(4096,-4441){\makebox(0,0)[lb]{\smash{{\SetFigFont{17}{20.4}{\rmdefault}{\mddefault}{\updefault}{\color[rgb]{.75,.38,0}Cell 3}%
}}}}
\put(5306,-2911){\makebox(0,0)[lb]{\smash{{\SetFigFont{17}{20.4}{\rmdefault}{\mddefault}{\updefault}{\color[rgb]{.75,.38,0}Cell 4}%
}}}}
\put(6726,-4351){\makebox(0,0)[lb]{\smash{{\SetFigFont{17}{20.4}{\rmdefault}{\mddefault}{\updefault}{\color[rgb]{.75,.38,0}Cell 6}%
}}}}
\put(2626,-4811){\makebox(0,0)[lb]{\smash{{\SetFigFont{17}{20.4}{\rmdefault}{\mddefault}{\updefault}{\color[rgb]{.75,.38,0}Cell 2}%
}}}}
\put(5636,-5021){\makebox(0,0)[lb]{\smash{{\SetFigFont{17}{20.4}{\rmdefault}{\mddefault}{\updefault}{\color[rgb]{0,.82,0}Route 3}%
}}}}
\put(4886,-5571){\makebox(0,0)[lb]{\smash{{\SetFigFont{17}{20.4}{\rmdefault}{\mddefault}{\updefault}{\color[rgb]{.75,.38,0}Cell 5}%
}}}}
\put(6826,-3061){\makebox(0,0)[lb]{\smash{{\SetFigFont{17}{20.4}{\rmdefault}{\mddefault}{\updefault}{\color[rgb]{0,.82,0}Route 1}%
}}}}
\put(2916,-3961){\makebox(0,0)[lb]{\smash{{\SetFigFont{17}{20.4}{\rmdefault}{\mddefault}{\updefault}{\color[rgb]{1,0,0}$\tau_{11}=t_{11}$}%
}}}}
\put(4108,-4186){\makebox(0,0)[lb]{\smash{{\SetFigFont{17}{20.4}{\rmdefault}{\mddefault}{\updefault}{\color[rgb]{1,0,0}$\tau_{12}=t_{12}$}%
}}}}
\put(5545,-4206){\makebox(0,0)[lb]{\smash{{\SetFigFont{17}{20.4}{\rmdefault}{\mddefault}{\updefault}{\color[rgb]{1,0,0}$t_{13}$}%
}}}}
\put(5944,-3823){\makebox(0,0)[lb]{\smash{{\SetFigFont{17}{20.4}{\rmdefault}{\mddefault}{\updefault}{\color[rgb]{.82,0,.82}$\tau_{13}$}%
}}}}
\put(3136,-5121){\makebox(0,0)[lb]{\smash{{\SetFigFont{17}{20.4}{\rmdefault}{\mddefault}{\updefault}{\color[rgb]{0,.82,0}Route 2}%
}}}}
\put(3001,-286){\makebox(0,0)[lb]{\smash{{\SetFigFont{14}{16.8}{\rmdefault}{\mddefault}{\updefault}{\color[rgb]{0,0,0}start}%
}}}}
\put(2476,-3511){\makebox(0,0)[lb]{\smash{{\SetFigFont{17}{20.4}{\rmdefault}{\mddefault}{\updefault}{\color[rgb]{.75,.38,0}Cell 1}%
}}}}
\put(4801,-286){\makebox(0,0)[lb]{\smash{{\SetFigFont{14}{16.8}{\rmdefault}{\mddefault}{\updefault}{\color[rgb]{0,0,0}handoff}%
}}}}
\put(5701,-286){\makebox(0,0)[lb]{\smash{{\SetFigFont{14}{16.8}{\rmdefault}{\mddefault}{\updefault}{\color[rgb]{0,0,0}terminate}%
}}}}
\end{picture}%

%% file: section4-analysis.tex
\section{Stationary User Distribution in Single-Route Network}  \label{single}

We first derive the stationary user distribution on a single route.  We construct a reference single-route memoryless network, where all the channel holding times are independently and exponentially distributed.  We prove insensitivity by showing an equivalence between the original network and the memoryless network in terms of stationary user distribution.

\subsection{Queueing Network Model for Single-Route Network}

\begin{figure}[ht]
\begin{center}
\scalebox{0.45}{ \input{queue01.pstex_t} }
\caption{Single-route network.}
\label{queue11}
\end{center}
\end{figure}
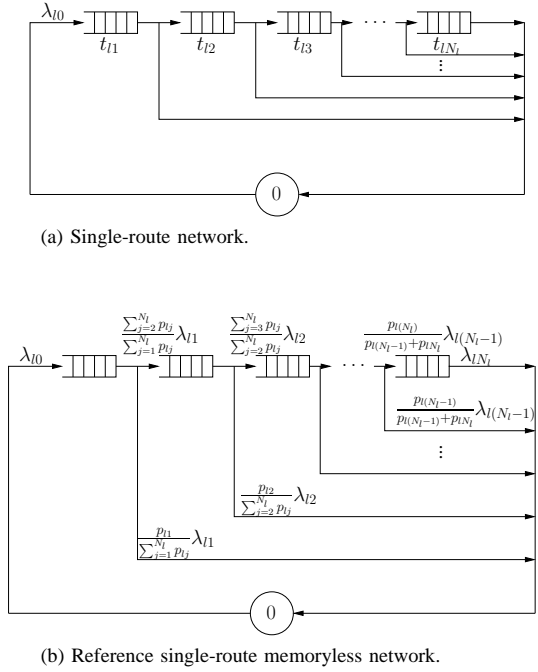

Consider exclusively the $l$th route in the network.  Throughout this section, we will carry the route index $l$ in most symbols, since they will be re-used in the analysis of multiple-route networks.

As shown in Fig.~\ref{queue11}(a), we model the route as a tandem-liked queueing network, except with early exists.  The node labeled with $0$ represents the exogenous world.  The $j$th queue, $1\leq j\leq N_l$, represents the $j$th stage of the route, and units in this queue represent sessions in the $j$th stage. Each queue has infinite servers, since the sessions are served in parallel with no waiting\footnote{Users move into and out of each cell in parallel.  Therefore, when considering the channel holding time as the service time of a queue that models mobility, this is equivalent to all users being served at the same time by its own dedicated server, which is the same as having infinite servers.  In terms of active user sessions, this model is accurate for communication systems with no admission control (e.g., WiFi) and gives reasonable approximation to systems with many available channels.}.

 The channel holding time of a session in the $j$th stage, $t_{lj}$, is equivalent to the service time of the $j$th queue.  The handoff of a session from the $j$th stage to the $(j+1)$th stage is equivalent to a unit movement from the $j$th queue to the $(j+1)$th queue. The termination of a session is equivalent to the movement from a queue to node $0$.

Let $p_{lk}$ denote the probability that a session lasts for $k$  stages. It is given by
\[
p_{lk}=P\Big[\sum_{j=1}^{k-1} \tau_{lj}< T_l\leq \sum_{j=1}^{k}\tau_{lj}\Big], \text{ for } 2\leq k \leq N_l-1 ,
\]
with $p_{l1}=P\left[ T_l\leq \tau_{l1}\right]$ and $p_{lN_l}=P\left[\sum_{j=1}^{N_l-1} \tau_{lj}< T_l\right]$.
Note that we have $\sum_{k=1}^{N_l}p_{lk}=1$. Given a session in the $k$th stage, it enters the $(k+1)$th stage with probability $\frac{\sum_{j=k+1}^{N_l}p_{lj}}{\sum_{j=k}^{N_l}p_{lj}}$ and terminates with probability $\frac{p_{lk}}{\sum_{j=k}^{N_l}p_{lj}}$.

\subsection{Reference Single-Route Memoryless Network}

We define a reference \emph{single-route memoryless network}, as a Jackson network with the same topology as the original single-route network, where each queue has infinitely many independent and exponential servers.  An illustration is shown in Fig.~\ref{queue11}(b).
By matching the mean service times in this memoryless network with those of the original network, we see that its external arrival rate is $\lambda_{l0}$, the service rate of the $j$th queue is  $\lambda_{lj}=\frac{1}{\overline{t}_{lj}}$. The routing probability from the $k$th queue to the $(k+1)$th queue is the probability that a session enters the $(k+1)$th stage conditioned on it is in the $k$th stage, $\frac{\sum_{j=k+1}^{N_l}p_{lj}}{\sum_{j=k}^{N_l}p_{lj}}$. The routing probability from the $k$th queue to node $0$ is $\frac{p_{lk}}{\sum_{j=k}^{N_l}p_{lj}}$.
Thus, the service rate from the $k$th queue to the $(k+1)$th queue is  $\frac{\sum_{j=k+1}^{N_l}p_{lj}}{\sum_{j=k}^{N_l}p_{lj}}\lambda_{lk}$, and the service rate from the $k$th queue to node $0$ is  $\frac{p_{lk}}{\sum_{j=k}^{N_l}p_{lj}}\lambda_{lk}$.

Let $w_{lj}'$ denote the positive invariant measure of the $j$th queue that satisfies the routing balance equations of the single-route memoryless network. $w_0'$ is the positive invariant measure of the node $0$. We adopt the convention that   $w_0'=1$.  It can be derived from the topology of Fig.~\ref{queue11}(b) that
\begin{align}
w_{0}'=&\lambda_{l0}w_{l1}',\\
\frac{\sum_{n=j}^{N_l}p_{ln}}{\sum_{n=j-1}^{N_l}p_{ln}}w_{lj-1}'=&w_{lj}', \quad 2 \leq j \leq N_l,
\end{align}
which leads to
\begin{align}
w_{l1}'=&\lambda_{l0},\\
w_{lj}'=&\lambda_{l0} (1- \sum_{n=1}^{j-1} p_{ln}), \quad 2 \leq j \leq N_l.
\end{align}
Because each queue has infinite servers, the departure intensity at the $j$th queue is $\lambda_{lj}x_{lj}$ when there are $x_{lj}$ users in it. Let $w_{lj} = \frac{w_{lj}'}{\lambda_{lj}}$.
Then the stationary user distribution w.r.t. $\mathbf{x}$ of this network is \cite{Book-StochasticNetwork}
\begin{equation}
\pi_0(\mathbf{x})=\prod_{j=1}^{N_l} e^{-w_{lj}} w_{lj}^{x_{lj}}\frac{1}{x_{lj}!}.
\end{equation}

\subsection{Insensitivity of Single-Route Network}

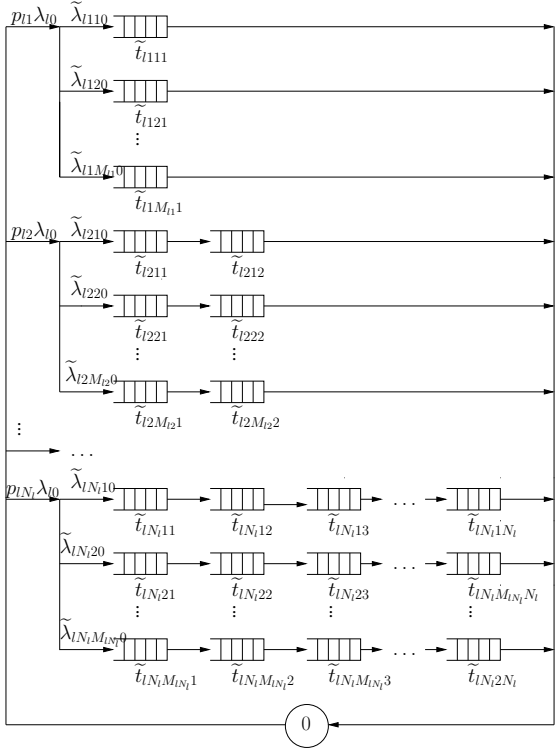
\begin{figure}[t]
\begin{center}
\scalebox{0.45}{ \input{queue03.pstex_t} }
\caption{Decoupled network.}
\label{queue13}
\end{center}
\end{figure}

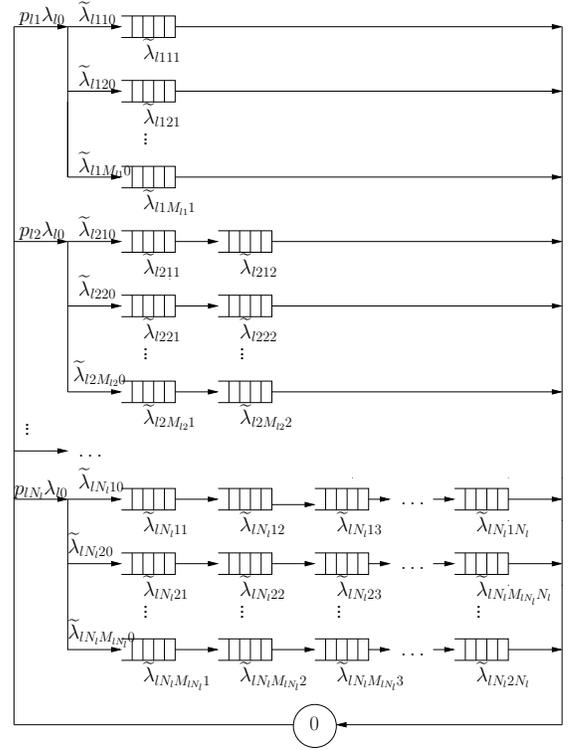
\begin{figure}[t]
\begin{center}
\scalebox{0.45}{ \input{queue04.pstex_t} }
\caption{Reference memoryless decoupled network.}
\label{queue14}
\end{center}
\end{figure}

For the original single route network, we employ a decomposition-composition approach to derive its stationary user distribution.

Given that one session lasts for $k$ stages, we denote the channel holding times as a $k$-dimensional random vector $\mathbf{\widehat{t}_{lk}}=\{\widehat{t}_{lk1}, \ldots \widehat{t}_{lkj}, \ldots ,\widehat{t}_{lkk}\}$, where $\widehat{t}_{lkj}$ is the channel holding time at the $j$th stage.  We assume that $\mathbf{\widehat{t}_{lk}}$ is an arbitrarily distributed discrete random vector
with $M_{lk}$ possible realizations\footnote{For a vector of continuous channel holding times, we can use a sequence of discrete distributions with decreasing granularity to approach its distribution.}.  For any $i$, $1\leq i\leq M_{lk}$, we define a $k$-dimensional deterministic vector $\mathbf{\widetilde{t}_{lki}}=[\widetilde{t}_{lki1},\ldots, \widetilde{t}_{lkij},\ldots, \widetilde{t}_{lkik}]^T$ corresponding to the $i$th realization of $\mathbf{\widehat{t}_{lk}}$. Let $q_{lki}$ be the probability of the $i$th realization given that the session lasts for $k$ stages.
Also, let $P_{lki}=p_{lk}q_{lki}$ denote the probability that a session lasts for $k$ stages and it is in the $i$th realization.

By doing so, we decompose the original network into a multiple-branch queueing network as shown in Fig. \ref{queue13}, which is referred to as the \emph{decoupled network}. In this network, there are $N_l$ main branches, where the $k$th main branch represents the event that a session lasts for $k$ stages.  The $k$th main branch contains $M_{lk}$ sub-branches, where the $i$th sub-branch  represents the realization where $\mathbf{\widehat{t}_{lk}}=\mathbf{\widetilde{t}_{lki}}$.  Furthermore, the $j$th queue in the $i$th sub-branch of the $k$th main branch represents the $j$th stage of the $i$th realization of the sessions that last for $k$ stages.

Hence, each queue of the decoupled network has infinite servers with \emph{deterministic} service time, $\widetilde{t}_{lkij}$, for the $j$th stage of the $i$th sub-branch of the $k$th main branch.  Furthermore, the arrival rate of the $i$th sub-branch of the $k$th main branch is $\widetilde{\lambda}_{lki0}= P_{lki} \lambda_{l0}$.
Let $\widetilde{\mathbf{x}}=[\{\widetilde{x}_{lkij}: 1\leq k\leq N_l, 1\leq j\leq k, 1\leq i\leq M_{lk}\}]^T$ be the vector of number of sessions in the $j$th stage of the $i$th sub-branch of the $k$th main branch.  Denote by $\pi_D(\widetilde{\mathbf{x}})$ the stationary user distribution of the decoupled network.

Note that the stationary distribution of a Jackson network with infinite servers at each queue is insensitive with respect to the distribution of the service times \cite{Insensitivity1}.  Therefore, $\pi_D(\widetilde{\mathbf{x}})$ remains unchanged if we create a reference Jackson network by replacing each queue with \emph{deterministic service time} in the decoupled network with a queue that has \emph{exponential distributed memoryless service time} with the same mean (e.g., the service rate at the $j$th queue of the $i$th sub-branch of the $k$th main branch $\widetilde{\lambda}_{lkij}=\frac{1}{\widetilde{t}_{lkij}}$), as shown in Fig. \ref{queue14}, which is referred to as the \emph{reference memoryless decoupled network}.

Let $\widetilde{w}_{lkij}'$ be the positive invariant measure of the $j$th queue of the $i$th sub-branch of the $k$th main branch of the reference memoryless decoupled network, which satisfies the routing balance equations with the convention that at node $0$, $w_0'=1$.  Since each sub-branch is a chain network, we have
\begin{eqnarray}\label{fomula_wkij_0}
\widetilde{w}_{lkij}'=P_{lki}\lambda_{l0}.
\end{eqnarray}
Let $\widetilde{w}_{lkij} = \frac{\widetilde{w}_{lkij}'}{\widetilde{\lambda}_{lkij}}$.  Then the stationary user distribution of the decoupled network as well as the reference memoryless decoupled network is
\begin{align}
\pi_D(\widetilde{\mathbf{x}})=\prod_{j=1}^{N_l}\prod_{k=j}^{N_l}\prod_{i=1}^{M_{lk}}e^{-\widetilde{w}_{lkij}}\widetilde{w}_{lkij}^{\widetilde{x}_{lkij}}\frac{1}{\widetilde{x}_{lkij}!}.
\end{align}

The stationary user distribution of the original single route network, $\pi(\mathbf{x})$, is the sum of $\pi_D(\widetilde{\mathbf{x}})$ satisfying $x_{lj}=\sum_{k=j}^{N_l}\sum_{i=1}^{M_{lk}}\widetilde{x}_{lkij}$, $\forall j$.
 To derive $\pi(\mathbf{x})$, we first introduce the following lemma.
\begin{lemma} \label{thm:JacksonComposition}
Consider a stationary open Jackson network with $N$ queues each with an infinite number of servers.
Let $x_j$ be the number of units in the $j$th queue and $\mathbf{x}=[x_1,\ldots x_N]^T$.
Suppose $\{\mathcal{J}_1, \mathcal{J}_2, \ldots \mathcal{J}_M\}$ is a set of mutually exclusive subsets of $\{1,2, \ldots, N\}$.  Let $z_i=\sum_{j\in \mathcal{J}_i}x_j$, $i=1,2,\ldots, M$, denoting the sum of units in the queues inside $\mathcal{J}_i$.  Then, the distribution of $\mathbf{z}=[z_1,\ldots z_M]^T$ is
\begin{align}
\pi(\mathbf{z})=\prod_{i=1}^{M}e^{-v_i}v_i^{z_i}\frac{1}{z_i!},
\end{align}
where $v_i=\sum_{j\in \mathcal{J}_i}w_j$, and $w_j$ is the expected number of units in the $j$th queue.
\end{lemma}
\begin{proof}
For a Jackson network with infinite servers at each queue, the stationary queue lengths are independent Poisson random variables with mean $w_j$ for the $j$th queue.  Hence, $z_i$ is Poisson with mean $v_i=\sum_{j\in \mathcal{J}_i}w_j$ for all $i$.  Furthermore, since $\{\mathcal{J}_i\}$ are mutually exclusive, $\{z_i\}$ are independent.
\end{proof}

Next, we note that the expected service time spent in the $j$th stage given that the $j$th stage exists, i.e., $j\leq k$ for the $k$th main branch, can be computed as
\begin{align}\label{fomula_t_j}
\nonumber\overline{t}_{lj}&=\frac{\sum_{k=j}^{N_l}\sum_{i=1}^{M_{lk}}P_{lki}\widetilde{t}_{lkij}}{\sum_{k=j}^{N_l}\sum_{i=1}^{M_{lk}}P_{lki}}\\
&=\frac{\sum_{k=j}^{N_l}\sum_{i=1}^{M_{lk}}P_{lki}\widetilde{t}_{lkij}}{1- \sum_{n=1}^{j-1} p_{ln}}.
\end{align}
Combining this with \eqref{fomula_wkij_0}, we have
\begin{align}\label{fomula_wkij}
\nonumber\sum_{k=j}^{N_l}\sum_{i=1}^{M_{lk}}\widetilde{w}_{lkij}&=\sum_{k=j}^{N_l}\sum_{i=1}^{M_{lk}}\frac{\lambda_{l0} P_{lki}}{\widetilde{\lambda}_{lkij}}\\
\nonumber&=\sum_{k=j}^{N_l}\sum_{i=1}^{M_{lk}}\lambda_{l0} P_{lki}\widetilde{t}_{lkij}\\
\nonumber&=\lambda_{l0}(1-\sum_{n=1}^{j-1} p_{ln})\overline{t}_{lj}\\
\nonumber&=\frac{\lambda_{l0}}{\lambda_{lj}}(1-\sum_{n=1}^{j-1} p_{ln})\\
&=w_{lj}.
\end{align}

Therefore, by Lemma \ref{thm:JacksonComposition}, we have
\begin{align}
\nonumber\pi(\mathbf{x})&=\sum_{\widetilde{\mathbf{x}}: x_{lj}=\sum_{k=j}^{N_{l}}\sum_{i=1}^{M_{lk}}\widetilde{x}_{lkij}, \forall j}\pi_{D}(\widetilde{\mathbf{x}})\\
&=\prod_{j=1}^{N_l}e^{-w_{lj}}\frac{w_{lj}^{x_{lj}}}{x_{lj}!},
\end{align}
which is restated as the following theorem:
\theorem
The single-route network has the same stationary user distribution as that of the corresponding single-route memoryless network: $\pi(\mathbf{x})=\pi_0(\mathbf{x})$.

\section{Stationary User Distribution in Multiple-Route Network}\label{multiple}
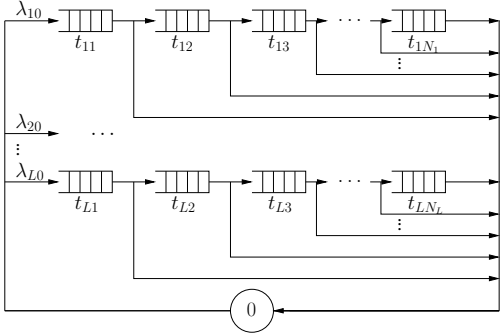
\begin{figure}[tbp]
\begin{center}
\scalebox{0.45}{ \input{queue2.pstex_t} }
\caption{Multiple-route network.}
\label{queue2}
\end{center}
\end{figure}

In this section, we study the general case with multiple routes.  We first extend the results from the previous section to show $\pi(\mathbf{x})=\pi_{0}(\mathbf{x})$ in a multiple-route network.  We then derive the stationary user distribution $\pi_1(\mathbf{y})$ with respect to cells and show its insensitivity.

\subsection {Queueing Network Model for Multiple-Route Network}

Since the $L$ routes are independent, we model the multiple-route network as a paralleling of $L$ single-route networks, as shown in Fig.~\ref{queue2}.
Similar to Section \ref{single}, we consider a reference multiple-route memoryless network, which is a paralleling of $L$ corresponding single-route memoryless networks. Then, we construct the decoupled multiple-route network, which is a paralleling of $L$ corresponding single-route decoupled networks.

\subsection{Insensitivity of $\pi(\mathbf{x})$}

\theorem
The multiple-route network has the same stationary user distribution as that of the corresponding multiple-route memoryless network.

\begin{proof}
\setcounter{paragraph}{0}
Since the routes are independent, the stationary user distribution of the multiple-route network can be computed as the product of the stationary user distribution of single-route networks:
\begin{align}\label{fomula_pi}
\pi(\mathbf{x})=\prod_{l=1}^{L}\prod_{j=1}^{N_l} e^{-w_{lj}}w_{lj}^{x_{lj}}\frac{1}{x_{lj}!} .
\end{align}
Since the same holds for the multiple-route memoryless network, we have $\pi(\mathbf{x})=\pi_{0}(\mathbf{x})$.
\end{proof}

\subsection{Insensitivity of $\pi_1(\mathbf{y})$}

Let $\overline{\lambda}_n$ be the average total arrival rate to cell $n$, including both new and handoff arrivals.  Let $\overline{\mathfrak{t}}_n$ be the average channel holding time in cell $n$, considering all routes and stages.  Thus,
\begin{align}
\overline{\lambda}_n&=\sum_{l,j:c(l,j)=n}\sum_{k=j}^{N_l}\sum_{i=1}^{M_{lk}}\lambda_{l0}P_{lki},\\
\overline{\mathfrak{t}}_n&=\frac{\sum_{l,j:c(l,j)=n}\sum_{k=j}^{N_l}\sum_{i=1}^{M_{lk}}\lambda_{l0}P_{lki}\widetilde{t}_{lkij}}{\sum_{l,j:c(l,j)=n}\sum_{k=j}^{N_l}\sum_{i=1}^{M_{lk}}\lambda_{l0}P_{lki}}.
\end{align}
Then from (\ref{fomula_wkij}), we have
\begin{align}
\nonumber\overline{\lambda}_n\overline{\mathfrak{t}}_n&=\sum_{l,j:c(l,j)=n} \sum_{k=j}^{N_l}\sum_{i=1}^{M_{lk}} \lambda_{l0}P_{lki}\widetilde{t}_{lkij}\\
\nonumber&=\sum_{l,j:c(l,j)=n} \sum_{k=j}^{N_l}\sum_{i=1}^{M_{lk}} \widetilde{w}_{lkij}\\
&=\sum_{l,j:c(l,j)=n} w_{lj}.
\end{align}

The joint stationary  user distribution among all cells can be computed as a summation over those entries of $\pi_0(\mathbf{x})$ satisfying $y_n=\sum_{l,j:c(l,j)=n} x_{lj}$, $\forall n$. Then from Lemma \ref{thm:JacksonComposition}, we obtain
\begin{align} \label{fomula_sd}
\nonumber \pi_{1}(\mathbf{y})&=\sum_{\mathbf{x}: y_n=\sum_{l,j:c(l,j)=n}x_{lj}, \forall n}\quad \prod_{l=1}^{L}\prod_{j=1}^{N_l} e^{-w_{lj}}w_{lj}^{x_{lj}}\frac{1}{x_{lj}!}\\
\nonumber&=\prod_n e^{-\left(\sum_{l,j:c(l,j)=n}  w_{lj}\right)}\left(\sum_{l,j:c(l,j)=n}  w_{lj}\right)^{y_{n}}\frac{1}{y_{n}!}\\
&=\prod_{n} e^{-\left(\overline{\lambda}_n\overline{\mathfrak{t}}_n\right)}\left(\overline{\lambda}_n\overline{\mathfrak{t}}_n\right)^{y_{n}}\frac{1}{y_{n}!}.
\end{align}

We make the following observations from \eqref{fomula_sd}:
\begin{itemize}
\item The number of users in each cell is independent and Poisson. This is in accordance with Theorem 9.27 in \cite{Book-StochasticNetwork}.
\item The stationary user distribution depends only on the average arrival rates and average channel holding times in individual cells, having the exact same form of an $M/M/\infty$ open Jackson network.  It is insensitive with respect to the distribution of channel holding times, or the correlation among them.  Furthermore, it is insensitive with respect to movement patterns, since the exact routing in the network is irrelevant.
\item The marginal distribution within a single cell depends only on the average arrival rate and average channel holding time at that cell. This useful property facilitates efficient system management and planning in practice, helping to avoid the need for collecting a large amount of user location data.

\end{itemize}

%% file: queue01.pstex_t
\begin{picture}(0,0)%
\includegraphics{queue01.pstex}%
\end{picture}%
\setlength{\unitlength}{3947sp}%
\begingroup\makeatletter\ifx\SetFigFont\undefined%
\gdef\SetFigFont#1#2#3#4#5{%
  \reset@font\fontsize{#1}{#2pt}%
  \fontfamily{#3}\fontseries{#4}\fontshape{#5}%
  \selectfont}%
\fi\endgroup%
\begin{picture}(7374,9322)(139,-9077)
\put(6151,-736){\makebox(0,0)[lb]{\smash{{\SetFigFont{17}{20.4}{\rmdefault}{\mddefault}{\updefault}{\color[rgb]{0,0,0}$\vdots$}%
}}}}
\put(5101,-61){\makebox(0,0)[lb]{\smash{{\SetFigFont{17}{20.4}{\rmdefault}{\mddefault}{\updefault}{\color[rgb]{0,0,0}$\ldots$}%
}}}}
\put(601,-3136){\makebox(0,0)[lb]{\smash{{\SetFigFont{17}{20.4}{\rmdefault}{\mddefault}{\updefault}{\color[rgb]{0,0,0}(a) Single-route network.}%
}}}}
\put(601, 14){\makebox(0,0)[lb]{\smash{{\SetFigFont{17}{20.4}{\rmdefault}{\mddefault}{\updefault}{\color[rgb]{0,0,0}$\lambda_{l0}$}%
}}}}
\put(1426,-436){\makebox(0,0)[lb]{\smash{{\SetFigFont{17}{20.4}{\rmdefault}{\mddefault}{\updefault}{\color[rgb]{0,0,0}$t_{l1}$}%
}}}}
\put(2776,-436){\makebox(0,0)[lb]{\smash{{\SetFigFont{17}{20.4}{\rmdefault}{\mddefault}{\updefault}{\color[rgb]{0,0,0}$t_{l2}$}%
}}}}
\put(4126,-436){\makebox(0,0)[lb]{\smash{{\SetFigFont{17}{20.4}{\rmdefault}{\mddefault}{\updefault}{\color[rgb]{0,0,0}$t_{l3}$}%
}}}}
\put(6076,-436){\makebox(0,0)[lb]{\smash{{\SetFigFont{17}{20.4}{\rmdefault}{\mddefault}{\updefault}{\color[rgb]{0,0,0}$t_{lN_l}$}%
}}}}
\put(3826,-2536){\makebox(0,0)[lb]{\smash{{\SetFigFont{17}{20.4}{\rmdefault}{\mddefault}{\updefault}{\color[rgb]{0,0,0}$0$}%
}}}}
\put(601,-8986){\makebox(0,0)[lb]{\smash{{\SetFigFont{17}{20.4}{\rmdefault}{\mddefault}{\updefault}{\color[rgb]{0,0,0}(b) Reference single-route memoryless network.}%
}}}}
\put(4801,-4861){\makebox(0,0)[lb]{\smash{{\SetFigFont{17}{20.4}{\rmdefault}{\mddefault}{\updefault}{\color[rgb]{0,0,0}$\ldots$}%
}}}}
\put(301,-4786){\makebox(0,0)[lb]{\smash{{\SetFigFont{17}{20.4}{\rmdefault}{\mddefault}{\updefault}{\color[rgb]{0,0,0}$\lambda_{l0}$}%
}}}}
\put(1951,-7336){\makebox(0,0)[lb]{\smash{{\SetFigFont{17}{20.4}{\rmdefault}{\mddefault}{\updefault}{\color[rgb]{0,0,0}$\frac{p_{l1}}{\sum_{j=1}^{N_l}p_{lj}}\lambda_{l1}$}%
}}}}
\put(3376,-6736){\makebox(0,0)[lb]{\smash{{\SetFigFont{17}{20.4}{\rmdefault}{\mddefault}{\updefault}{\color[rgb]{0,0,0}$\frac{p_{l2}}{\sum_{j=2}^{N_l}p_{lj}}\lambda_{l2}$}%
}}}}
\put(3751,-8386){\makebox(0,0)[lb]{\smash{{\SetFigFont{17}{20.4}{\rmdefault}{\mddefault}{\updefault}{\color[rgb]{0,0,0}$0$}%
}}}}
\put(6451,-4786){\makebox(0,0)[lb]{\smash{{\SetFigFont{17}{20.4}{\rmdefault}{\mddefault}{\updefault}{\color[rgb]{0,0,0}$\lambda_{lN_l}$}%
}}}}
\put(5101,-4486){\makebox(0,0)[lb]{\smash{{\SetFigFont{17}{20.4}{\rmdefault}{\mddefault}{\updefault}{\color[rgb]{0,0,0}$\frac{p_{l(N_l)}}{p_{l(N_l-1)}+p_{lN_l}}\lambda_{l(N_l-1)}$}%
}}}}
\put(1726,-4486){\makebox(0,0)[lb]{\smash{{\SetFigFont{17}{20.4}{\rmdefault}{\mddefault}{\updefault}{\color[rgb]{0,0,0}$\frac{\sum_{j=2}^{N_l}p_{lj}}{\sum_{j=1}^{N_l}p_{lj}}\lambda_{l1}$}%
}}}}
\put(3226,-4486){\makebox(0,0)[lb]{\smash{{\SetFigFont{17}{20.4}{\rmdefault}{\mddefault}{\updefault}{\color[rgb]{0,0,0}$\frac{\sum_{j=3}^{N_l}p_{lj}}{\sum_{j=2}^{N_l}p_{lj}}\lambda_{l2}$}%
}}}}
\put(5551,-5536){\makebox(0,0)[lb]{\smash{{\SetFigFont{17}{20.4}{\rmdefault}{\mddefault}{\updefault}{\color[rgb]{0,0,0}$\frac{p_{l(N_l-1)}}{p_{l(N_l-1)}+p_{lN_l}}\lambda_{l(N_l-1)}$}%
}}}}
\put(6151,-6136){\makebox(0,0)[lb]{\smash{{\SetFigFont{17}{20.4}{\rmdefault}{\mddefault}{\updefault}{\color[rgb]{0,0,0}$\vdots$}%
}}}}
\end{picture}%

%% file: queue03.pstex_t
\begin{picture}(0,0)%
\includegraphics{queue03.pstex}%
\end{picture}%
\setlength{\unitlength}{3947sp}%
\begingroup\makeatletter\ifx\SetFigFont\undefined%
\gdef\SetFigFont#1#2#3#4#5{%
  \reset@font\fontsize{#1}{#2pt}%
  \fontfamily{#3}\fontseries{#4}\fontshape{#5}%
  \selectfont}%
\fi\endgroup%
\begin{picture}(7677,10380)(736,-10135)
\put(2551,-8011){\makebox(0,0)[lb]{\smash{{\SetFigFont{17}{20.4}{\rmdefault}{\mddefault}{\updefault}{\color[rgb]{0,0,0}$\widetilde{t}_{lN_l21}$}%
}}}}
\put(1651, 14){\makebox(0,0)[lb]{\smash{{\SetFigFont{17}{20.4}{\rmdefault}{\mddefault}{\updefault}{\color[rgb]{0,0,0}$\widetilde{\lambda}_{l110}$}%
}}}}
\put(826, 14){\makebox(0,0)[lb]{\smash{{\SetFigFont{17}{20.4}{\rmdefault}{\mddefault}{\updefault}{\color[rgb]{0,0,0}$p_{l1}\lambda_{l0}$}%
}}}}
\put(2551,-5611){\makebox(0,0)[lb]{\smash{{\SetFigFont{17}{20.4}{\rmdefault}{\mddefault}{\updefault}{\color[rgb]{0,0,0}$\widetilde{t}_{l2M_{l2}1}$}%
}}}}
\put(3901,-5611){\makebox(0,0)[lb]{\smash{{\SetFigFont{17}{20.4}{\rmdefault}{\mddefault}{\updefault}{\color[rgb]{0,0,0}$\widetilde{t}_{l2M_{l2}2}$}%
}}}}
\put(3901,-4711){\makebox(0,0)[lb]{\smash{{\SetFigFont{17}{20.4}{\rmdefault}{\mddefault}{\updefault}{\color[rgb]{0,0,0}$\vdots$}%
}}}}
\put(2551,-4711){\makebox(0,0)[lb]{\smash{{\SetFigFont{17}{20.4}{\rmdefault}{\mddefault}{\updefault}{\color[rgb]{0,0,0}$\vdots$}%
}}}}
\put(2551,-4411){\makebox(0,0)[lb]{\smash{{\SetFigFont{17}{20.4}{\rmdefault}{\mddefault}{\updefault}{\color[rgb]{0,0,0}$\widetilde{t}_{l221}$}%
}}}}
\put(3901,-4411){\makebox(0,0)[lb]{\smash{{\SetFigFont{17}{20.4}{\rmdefault}{\mddefault}{\updefault}{\color[rgb]{0,0,0}$\widetilde{t}_{l222}$}%
}}}}
\put(1651,-3811){\makebox(0,0)[lb]{\smash{{\SetFigFont{17}{20.4}{\rmdefault}{\mddefault}{\updefault}{\color[rgb]{0,0,0}$\widetilde{\lambda}_{l220}$}%
}}}}
\put(2551,-3511){\makebox(0,0)[lb]{\smash{{\SetFigFont{17}{20.4}{\rmdefault}{\mddefault}{\updefault}{\color[rgb]{0,0,0}$\widetilde{t}_{l211}$}%
}}}}
\put(3901,-3511){\makebox(0,0)[lb]{\smash{{\SetFigFont{17}{20.4}{\rmdefault}{\mddefault}{\updefault}{\color[rgb]{0,0,0}$\widetilde{t}_{l212}$}%
}}}}
\put(826,-2986){\makebox(0,0)[lb]{\smash{{\SetFigFont{17}{20.4}{\rmdefault}{\mddefault}{\updefault}{\color[rgb]{0,0,0}$p_{l2}\lambda_{l0}$}%
}}}}
\put(2551,-1711){\makebox(0,0)[lb]{\smash{{\SetFigFont{17}{20.4}{\rmdefault}{\mddefault}{\updefault}{\color[rgb]{0,0,0}$\vdots$}%
}}}}
\put(1651,-886){\makebox(0,0)[lb]{\smash{{\SetFigFont{17}{20.4}{\rmdefault}{\mddefault}{\updefault}{\color[rgb]{0,0,0}$\widetilde{\lambda}_{l120}$}%
}}}}
\put(1651,-2986){\makebox(0,0)[lb]{\smash{{\SetFigFont{17}{20.4}{\rmdefault}{\mddefault}{\updefault}{\color[rgb]{0,0,0}$\widetilde{\lambda}_{l210}$}%
}}}}
\put(1651,-2086){\makebox(0,0)[lb]{\smash{{\SetFigFont{17}{20.4}{\rmdefault}{\mddefault}{\updefault}{\color[rgb]{0,0,0}$\widetilde{\lambda}_{l1M_{l1}0}$}%
}}}}
\put(2551,-2611){\makebox(0,0)[lb]{\smash{{\SetFigFont{17}{20.4}{\rmdefault}{\mddefault}{\updefault}{\color[rgb]{0,0,0}$\widetilde{t}_{l1M_{l1}1}$}%
}}}}
\put(2551,-511){\makebox(0,0)[lb]{\smash{{\SetFigFont{17}{20.4}{\rmdefault}{\mddefault}{\updefault}{\color[rgb]{0,0,0}$\widetilde{t}_{l111}$}%
}}}}
\put(2551,-1411){\makebox(0,0)[lb]{\smash{{\SetFigFont{17}{20.4}{\rmdefault}{\mddefault}{\updefault}{\color[rgb]{0,0,0}$\widetilde{t}_{l121}$}%
}}}}
\put(4876,-9886){\makebox(0,0)[lb]{\smash{{\SetFigFont{17}{20.4}{\rmdefault}{\mddefault}{\updefault}{\color[rgb]{0,0,0}$0$}%
}}}}
\put(6151,-8836){\makebox(0,0)[lb]{\smash{{\SetFigFont{17}{20.4}{\rmdefault}{\mddefault}{\updefault}{\color[rgb]{0,0,0}$\ldots$}%
}}}}
\put(2551,-9211){\makebox(0,0)[lb]{\smash{{\SetFigFont{17}{20.4}{\rmdefault}{\mddefault}{\updefault}{\color[rgb]{0,0,0}$\widetilde{t}_{lN_lM_{lN_l}1}$}%
}}}}
\put(7201,-9211){\makebox(0,0)[lb]{\smash{{\SetFigFont{17}{20.4}{\rmdefault}{\mddefault}{\updefault}{\color[rgb]{0,0,0}$\widetilde{t}_{lN_l2N_l}$}%
}}}}
\put(2551,-8311){\makebox(0,0)[lb]{\smash{{\SetFigFont{17}{20.4}{\rmdefault}{\mddefault}{\updefault}{\color[rgb]{0,0,0}$\vdots$}%
}}}}
\put(3901,-9211){\makebox(0,0)[lb]{\smash{{\SetFigFont{17}{20.4}{\rmdefault}{\mddefault}{\updefault}{\color[rgb]{0,0,0}$\widetilde{t}_{lN_lM_{lN_l}2}$}%
}}}}
\put(5251,-9211){\makebox(0,0)[lb]{\smash{{\SetFigFont{17}{20.4}{\rmdefault}{\mddefault}{\updefault}{\color[rgb]{0,0,0}$\widetilde{t}_{lN_lM_{lN_l}3}$}%
}}}}
\put(3901,-8011){\makebox(0,0)[lb]{\smash{{\SetFigFont{17}{20.4}{\rmdefault}{\mddefault}{\updefault}{\color[rgb]{0,0,0}$\widetilde{t}_{lN_l22}$}%
}}}}
\put(5251,-8011){\makebox(0,0)[lb]{\smash{{\SetFigFont{17}{20.4}{\rmdefault}{\mddefault}{\updefault}{\color[rgb]{0,0,0}$\widetilde{t}_{lN_l23}$}%
}}}}
\put(7201,-8011){\makebox(0,0)[lb]{\smash{{\SetFigFont{17}{20.4}{\rmdefault}{\mddefault}{\updefault}{\color[rgb]{0,0,0}$\widetilde{t}_{lN_lM_{lN_l}N_l}$}%
}}}}
\put(7201,-7111){\makebox(0,0)[lb]{\smash{{\SetFigFont{17}{20.4}{\rmdefault}{\mddefault}{\updefault}{\color[rgb]{0,0,0}$\widetilde{t}_{lN_l1N_l}$}%
}}}}
\put(7201,-8311){\makebox(0,0)[lb]{\smash{{\SetFigFont{17}{20.4}{\rmdefault}{\mddefault}{\updefault}{\color[rgb]{0,0,0}$\vdots$}%
}}}}
\put(5251,-8311){\makebox(0,0)[lb]{\smash{{\SetFigFont{17}{20.4}{\rmdefault}{\mddefault}{\updefault}{\color[rgb]{0,0,0}$\vdots$}%
}}}}
\put(3901,-8311){\makebox(0,0)[lb]{\smash{{\SetFigFont{17}{20.4}{\rmdefault}{\mddefault}{\updefault}{\color[rgb]{0,0,0}$\vdots$}%
}}}}
\put(6151,-7636){\makebox(0,0)[lb]{\smash{{\SetFigFont{17}{20.4}{\rmdefault}{\mddefault}{\updefault}{\color[rgb]{0,0,0}$\ldots$}%
}}}}
\put(3901,-7111){\makebox(0,0)[lb]{\smash{{\SetFigFont{17}{20.4}{\rmdefault}{\mddefault}{\updefault}{\color[rgb]{0,0,0}$\widetilde{t}_{lN_l12}$}%
}}}}
\put(5251,-7111){\makebox(0,0)[lb]{\smash{{\SetFigFont{17}{20.4}{\rmdefault}{\mddefault}{\updefault}{\color[rgb]{0,0,0}$\widetilde{t}_{lN_l13}$}%
}}}}
\put(6151,-6736){\makebox(0,0)[lb]{\smash{{\SetFigFont{17}{20.4}{\rmdefault}{\mddefault}{\updefault}{\color[rgb]{0,0,0}$\ldots$}%
}}}}
\put(901,-5761){\makebox(0,0)[lb]{\smash{{\SetFigFont{17}{20.4}{\rmdefault}{\mddefault}{\updefault}{\color[rgb]{0,0,0}$\vdots$}%
}}}}
\put(1651,-6061){\makebox(0,0)[lb]{\smash{{\SetFigFont{17}{20.4}{\rmdefault}{\mddefault}{\updefault}{\color[rgb]{0,0,0}$\ldots$}%
}}}}
\put(751,-6586){\makebox(0,0)[lb]{\smash{{\SetFigFont{17}{20.4}{\rmdefault}{\mddefault}{\updefault}{\color[rgb]{0,0,0}$p_{lN_l}\lambda_{l0}$}%
}}}}
\put(1651,-6511){\makebox(0,0)[lb]{\smash{{\SetFigFont{17}{20.4}{\rmdefault}{\mddefault}{\updefault}{\color[rgb]{0,0,0}$\widetilde{\lambda}_{lN_l10}$}%
}}}}
\put(1501,-8611){\makebox(0,0)[lb]{\smash{{\SetFigFont{17}{20.4}{\rmdefault}{\mddefault}{\updefault}{\color[rgb]{0,0,0}$\widetilde{\lambda}_{lN_lM_{lN_l}0}$}%
}}}}
\put(1501,-7411){\makebox(0,0)[lb]{\smash{{\SetFigFont{17}{20.4}{\rmdefault}{\mddefault}{\updefault}{\color[rgb]{0,0,0}$\widetilde{\lambda}_{lN_l20}$}%
}}}}
\put(1576,-5011){\makebox(0,0)[lb]{\smash{{\SetFigFont{17}{20.4}{\rmdefault}{\mddefault}{\updefault}{\color[rgb]{0,0,0}$\widetilde{\lambda}_{l2M_{l2}0}$}%
}}}}
\put(2551,-7111){\makebox(0,0)[lb]{\smash{{\SetFigFont{17}{20.4}{\rmdefault}{\mddefault}{\updefault}{\color[rgb]{0,0,0}$\widetilde{t}_{lN_l11}$}%
}}}}
\end{picture}%

%% file: queue04.pstex_t
\begin{picture}(0,0)%
\includegraphics{queue04.pstex}%
\end{picture}%
\setlength{\unitlength}{3947sp}%
\begingroup\makeatletter\ifx\SetFigFont\undefined%
\gdef\SetFigFont#1#2#3#4#5{%
  \reset@font\fontsize{#1}{#2pt}%
  \fontfamily{#3}\fontseries{#4}\fontshape{#5}%
  \selectfont}%
\fi\endgroup%
\begin{picture}(7677,10380)(661,-9910)
\put(7126,-7786){\makebox(0,0)[lb]{\smash{{\SetFigFont{17}{20.4}{\rmdefault}{\mddefault}{\updefault}{\color[rgb]{0,0,0}$\widetilde{\lambda}_{lN_lM_{lN_l}N_l}$}%
}}}}
\put(4801,-9661){\makebox(0,0)[lb]{\smash{{\SetFigFont{17}{20.4}{\rmdefault}{\mddefault}{\updefault}{\color[rgb]{0,0,0}$0$}%
}}}}
\put(1576,239){\makebox(0,0)[lb]{\smash{{\SetFigFont{17}{20.4}{\rmdefault}{\mddefault}{\updefault}{\color[rgb]{0,0,0}$\widetilde{\lambda}_{l110}$}%
}}}}
\put(751,239){\makebox(0,0)[lb]{\smash{{\SetFigFont{17}{20.4}{\rmdefault}{\mddefault}{\updefault}{\color[rgb]{0,0,0}$p_{l1}\lambda_{l0}$}%
}}}}
\put(2476,-5386){\makebox(0,0)[lb]{\smash{{\SetFigFont{17}{20.4}{\rmdefault}{\mddefault}{\updefault}{\color[rgb]{0,0,0}$\widetilde{\lambda}_{l2M_{l2}1}$}%
}}}}
\put(3826,-5386){\makebox(0,0)[lb]{\smash{{\SetFigFont{17}{20.4}{\rmdefault}{\mddefault}{\updefault}{\color[rgb]{0,0,0}$\widetilde{\lambda}_{l2M_{l2}2}$}%
}}}}
\put(3826,-4486){\makebox(0,0)[lb]{\smash{{\SetFigFont{17}{20.4}{\rmdefault}{\mddefault}{\updefault}{\color[rgb]{0,0,0}$\vdots$}%
}}}}
\put(2476,-4486){\makebox(0,0)[lb]{\smash{{\SetFigFont{17}{20.4}{\rmdefault}{\mddefault}{\updefault}{\color[rgb]{0,0,0}$\vdots$}%
}}}}
\put(2476,-4186){\makebox(0,0)[lb]{\smash{{\SetFigFont{17}{20.4}{\rmdefault}{\mddefault}{\updefault}{\color[rgb]{0,0,0}$\widetilde{\lambda}_{l221}$}%
}}}}
\put(3826,-4186){\makebox(0,0)[lb]{\smash{{\SetFigFont{17}{20.4}{\rmdefault}{\mddefault}{\updefault}{\color[rgb]{0,0,0}$\widetilde{\lambda}_{l222}$}%
}}}}
\put(1576,-3586){\makebox(0,0)[lb]{\smash{{\SetFigFont{17}{20.4}{\rmdefault}{\mddefault}{\updefault}{\color[rgb]{0,0,0}$\widetilde{\lambda}_{l220}$}%
}}}}
\put(2476,-3286){\makebox(0,0)[lb]{\smash{{\SetFigFont{17}{20.4}{\rmdefault}{\mddefault}{\updefault}{\color[rgb]{0,0,0}$\widetilde{\lambda}_{l211}$}%
}}}}
\put(3826,-3286){\makebox(0,0)[lb]{\smash{{\SetFigFont{17}{20.4}{\rmdefault}{\mddefault}{\updefault}{\color[rgb]{0,0,0}$\widetilde{\lambda}_{l212}$}%
}}}}
\put(751,-2761){\makebox(0,0)[lb]{\smash{{\SetFigFont{17}{20.4}{\rmdefault}{\mddefault}{\updefault}{\color[rgb]{0,0,0}$p_{l2}\lambda_{l0}$}%
}}}}
\put(2476,-1486){\makebox(0,0)[lb]{\smash{{\SetFigFont{17}{20.4}{\rmdefault}{\mddefault}{\updefault}{\color[rgb]{0,0,0}$\vdots$}%
}}}}
\put(1576,-661){\makebox(0,0)[lb]{\smash{{\SetFigFont{17}{20.4}{\rmdefault}{\mddefault}{\updefault}{\color[rgb]{0,0,0}$\widetilde{\lambda}_{l120}$}%
}}}}
\put(1576,-2761){\makebox(0,0)[lb]{\smash{{\SetFigFont{17}{20.4}{\rmdefault}{\mddefault}{\updefault}{\color[rgb]{0,0,0}$\widetilde{\lambda}_{l210}$}%
}}}}
\put(1576,-1861){\makebox(0,0)[lb]{\smash{{\SetFigFont{17}{20.4}{\rmdefault}{\mddefault}{\updefault}{\color[rgb]{0,0,0}$\widetilde{\lambda}_{l1M_{l1}0}$}%
}}}}
\put(2476,-2386){\makebox(0,0)[lb]{\smash{{\SetFigFont{17}{20.4}{\rmdefault}{\mddefault}{\updefault}{\color[rgb]{0,0,0}$\widetilde{\lambda}_{l1M_{l1}1}$}%
}}}}
\put(2476,-286){\makebox(0,0)[lb]{\smash{{\SetFigFont{17}{20.4}{\rmdefault}{\mddefault}{\updefault}{\color[rgb]{0,0,0}$\widetilde{\lambda}_{l111}$}%
}}}}
\put(2476,-1186){\makebox(0,0)[lb]{\smash{{\SetFigFont{17}{20.4}{\rmdefault}{\mddefault}{\updefault}{\color[rgb]{0,0,0}$\widetilde{\lambda}_{l121}$}%
}}}}
\put(6076,-8611){\makebox(0,0)[lb]{\smash{{\SetFigFont{17}{20.4}{\rmdefault}{\mddefault}{\updefault}{\color[rgb]{0,0,0}$\ldots$}%
}}}}
\put(2476,-8086){\makebox(0,0)[lb]{\smash{{\SetFigFont{17}{20.4}{\rmdefault}{\mddefault}{\updefault}{\color[rgb]{0,0,0}$\vdots$}%
}}}}
\put(5176,-8986){\makebox(0,0)[lb]{\smash{{\SetFigFont{17}{20.4}{\rmdefault}{\mddefault}{\updefault}{\color[rgb]{0,0,0}$\widetilde{\lambda}_{lN_lM_{lN_l}3}$}%
}}}}
\put(3826,-7786){\makebox(0,0)[lb]{\smash{{\SetFigFont{17}{20.4}{\rmdefault}{\mddefault}{\updefault}{\color[rgb]{0,0,0}$\widetilde{\lambda}_{lN_l22}$}%
}}}}
\put(5176,-7786){\makebox(0,0)[lb]{\smash{{\SetFigFont{17}{20.4}{\rmdefault}{\mddefault}{\updefault}{\color[rgb]{0,0,0}$\widetilde{\lambda}_{lN_l23}$}%
}}}}
\put(7126,-6886){\makebox(0,0)[lb]{\smash{{\SetFigFont{17}{20.4}{\rmdefault}{\mddefault}{\updefault}{\color[rgb]{0,0,0}$\widetilde{\lambda}_{lN_l1N_l}$}%
}}}}
\put(7126,-8086){\makebox(0,0)[lb]{\smash{{\SetFigFont{17}{20.4}{\rmdefault}{\mddefault}{\updefault}{\color[rgb]{0,0,0}$\vdots$}%
}}}}
\put(5176,-8086){\makebox(0,0)[lb]{\smash{{\SetFigFont{17}{20.4}{\rmdefault}{\mddefault}{\updefault}{\color[rgb]{0,0,0}$\vdots$}%
}}}}
\put(3826,-8086){\makebox(0,0)[lb]{\smash{{\SetFigFont{17}{20.4}{\rmdefault}{\mddefault}{\updefault}{\color[rgb]{0,0,0}$\vdots$}%
}}}}
\put(6076,-7411){\makebox(0,0)[lb]{\smash{{\SetFigFont{17}{20.4}{\rmdefault}{\mddefault}{\updefault}{\color[rgb]{0,0,0}$\ldots$}%
}}}}
\put(3826,-6886){\makebox(0,0)[lb]{\smash{{\SetFigFont{17}{20.4}{\rmdefault}{\mddefault}{\updefault}{\color[rgb]{0,0,0}$\widetilde{\lambda}_{lN_l12}$}%
}}}}
\put(5176,-6886){\makebox(0,0)[lb]{\smash{{\SetFigFont{17}{20.4}{\rmdefault}{\mddefault}{\updefault}{\color[rgb]{0,0,0}$\widetilde{\lambda}_{lN_l13}$}%
}}}}
\put(6076,-6511){\makebox(0,0)[lb]{\smash{{\SetFigFont{17}{20.4}{\rmdefault}{\mddefault}{\updefault}{\color[rgb]{0,0,0}$\ldots$}%
}}}}
\put(826,-5536){\makebox(0,0)[lb]{\smash{{\SetFigFont{17}{20.4}{\rmdefault}{\mddefault}{\updefault}{\color[rgb]{0,0,0}$\vdots$}%
}}}}
\put(1576,-5836){\makebox(0,0)[lb]{\smash{{\SetFigFont{17}{20.4}{\rmdefault}{\mddefault}{\updefault}{\color[rgb]{0,0,0}$\ldots$}%
}}}}
\put(676,-6361){\makebox(0,0)[lb]{\smash{{\SetFigFont{17}{20.4}{\rmdefault}{\mddefault}{\updefault}{\color[rgb]{0,0,0}$p_{lN_l}\lambda_{l0}$}%
}}}}
\put(1576,-6286){\makebox(0,0)[lb]{\smash{{\SetFigFont{17}{20.4}{\rmdefault}{\mddefault}{\updefault}{\color[rgb]{0,0,0}$\widetilde{\lambda}_{lN_l10}$}%
}}}}
\put(1426,-8386){\makebox(0,0)[lb]{\smash{{\SetFigFont{17}{20.4}{\rmdefault}{\mddefault}{\updefault}{\color[rgb]{0,0,0}$\widetilde{\lambda}_{lN_lM_{lN_l}0}$}%
}}}}
\put(1426,-7186){\makebox(0,0)[lb]{\smash{{\SetFigFont{17}{20.4}{\rmdefault}{\mddefault}{\updefault}{\color[rgb]{0,0,0}$\widetilde{\lambda}_{lN_l20}$}%
}}}}
\put(1501,-4786){\makebox(0,0)[lb]{\smash{{\SetFigFont{17}{20.4}{\rmdefault}{\mddefault}{\updefault}{\color[rgb]{0,0,0}$\widetilde{\lambda}_{l2M_{l2}0}$}%
}}}}
\put(2476,-6886){\makebox(0,0)[lb]{\smash{{\SetFigFont{17}{20.4}{\rmdefault}{\mddefault}{\updefault}{\color[rgb]{0,0,0}$\widetilde{\lambda}_{lN_l11}$}%
}}}}
\put(7126,-8986){\makebox(0,0)[lb]{\smash{{\SetFigFont{17}{20.4}{\rmdefault}{\mddefault}{\updefault}{\color[rgb]{0,0,0}$\widetilde{\lambda}_{lN_l2N_l}$}%
}}}}
\put(3826,-8986){\makebox(0,0)[lb]{\smash{{\SetFigFont{17}{20.4}{\rmdefault}{\mddefault}{\updefault}{\color[rgb]{0,0,0}$\widetilde{\lambda}_{lN_lM_{lN_l}2}$}%
}}}}
\put(2476,-7786){\makebox(0,0)[lb]{\smash{{\SetFigFont{17}{20.4}{\rmdefault}{\mddefault}{\updefault}{\color[rgb]{0,0,0}$\widetilde{\lambda}_{lN_l21}$}%
}}}}
\put(2476,-8986){\makebox(0,0)[lb]{\smash{{\SetFigFont{17}{20.4}{\rmdefault}{\mddefault}{\updefault}{\color[rgb]{0,0,0}$\widetilde{\lambda}_{lN_lM_{lN_l}1}$}%
}}}}
\end{picture}%

%% file: queue2.pstex_t
\begin{picture}(0,0)%
\includegraphics{queue2.pstex}%
\end{picture}%
\setlength{\unitlength}{3947sp}%
\begingroup\makeatletter\ifx\SetFigFont\undefined%
\gdef\SetFigFont#1#2#3#4#5{%
  \reset@font\fontsize{#1}{#2pt}%
  \fontfamily{#3}\fontseries{#4}\fontshape{#5}%
  \selectfont}%
\fi\endgroup%
\begin{picture}(6924,4663)(589,-3968)
\put(6226,-2236){\makebox(0,0)[lb]{\smash{{\SetFigFont{17}{20.4}{\rmdefault}{\mddefault}{\updefault}{\color[rgb]{0,0,0}$t_{LN_L}$}%
}}}}
\put(5251,389){\makebox(0,0)[lb]{\smash{{\SetFigFont{17}{20.4}{\rmdefault}{\mddefault}{\updefault}{\color[rgb]{0,0,0}$\ldots$}%
}}}}
\put(751,464){\makebox(0,0)[lb]{\smash{{\SetFigFont{17}{20.4}{\rmdefault}{\mddefault}{\updefault}{\color[rgb]{0,0,0}$\lambda_{10}$}%
}}}}
\put(751,-1786){\makebox(0,0)[lb]{\smash{{\SetFigFont{17}{20.4}{\rmdefault}{\mddefault}{\updefault}{\color[rgb]{0,0,0}$\lambda_{L0}$}%
}}}}
\put(751,-1111){\makebox(0,0)[lb]{\smash{{\SetFigFont{17}{20.4}{\rmdefault}{\mddefault}{\updefault}{\color[rgb]{0,0,0}$\lambda_{20}$}%
}}}}
\put(751,-1486){\makebox(0,0)[lb]{\smash{{\SetFigFont{17}{20.4}{\rmdefault}{\mddefault}{\updefault}{\color[rgb]{0,0,0}$\vdots$}%
}}}}
\put(1801,-1261){\makebox(0,0)[lb]{\smash{{\SetFigFont{17}{20.4}{\rmdefault}{\mddefault}{\updefault}{\color[rgb]{0,0,0}$\cdots$}%
}}}}
\put(6076,-2536){\makebox(0,0)[lb]{\smash{{\SetFigFont{17}{20.4}{\rmdefault}{\mddefault}{\updefault}{\color[rgb]{0,0,0}$\vdots$}%
}}}}
\put(6076,-286){\makebox(0,0)[lb]{\smash{{\SetFigFont{17}{20.4}{\rmdefault}{\mddefault}{\updefault}{\color[rgb]{0,0,0}$\vdots$}%
}}}}
\put(1576, 14){\makebox(0,0)[lb]{\smash{{\SetFigFont{17}{20.4}{\rmdefault}{\mddefault}{\updefault}{\color[rgb]{0,0,0}$t_{11}$}%
}}}}
\put(6226, 14){\makebox(0,0)[lb]{\smash{{\SetFigFont{17}{20.4}{\rmdefault}{\mddefault}{\updefault}{\color[rgb]{0,0,0}$t_{1N_1}$}%
}}}}
\put(4276, 14){\makebox(0,0)[lb]{\smash{{\SetFigFont{17}{20.4}{\rmdefault}{\mddefault}{\updefault}{\color[rgb]{0,0,0}$t_{13}$}%
}}}}
\put(2926, 14){\makebox(0,0)[lb]{\smash{{\SetFigFont{17}{20.4}{\rmdefault}{\mddefault}{\updefault}{\color[rgb]{0,0,0}$t_{12}$}%
}}}}
\put(5251,-1861){\makebox(0,0)[lb]{\smash{{\SetFigFont{17}{20.4}{\rmdefault}{\mddefault}{\updefault}{\color[rgb]{0,0,0}$\ldots$}%
}}}}
\put(3976,-3736){\makebox(0,0)[lb]{\smash{{\SetFigFont{17}{20.4}{\rmdefault}{\mddefault}{\updefault}{\color[rgb]{0,0,0}$0$}%
}}}}
\put(1576,-2236){\makebox(0,0)[lb]{\smash{{\SetFigFont{17}{20.4}{\rmdefault}{\mddefault}{\updefault}{\color[rgb]{0,0,0}$t_{L1}$}%
}}}}
\put(2926,-2236){\makebox(0,0)[lb]{\smash{{\SetFigFont{17}{20.4}{\rmdefault}{\mddefault}{\updefault}{\color[rgb]{0,0,0}$t_{L2}$}%
}}}}
\put(4276,-2236){\makebox(0,0)[lb]{\smash{{\SetFigFont{17}{20.4}{\rmdefault}{\mddefault}{\updefault}{\color[rgb]{0,0,0}$t_{L3}$}%
}}}}
\end{picture}%

%% file: section5-experiment.tex
\section {Experimental Study} \label{section_experiment}

In this section, our analysis is validated via experimenting with real-world traces. We first present the data source and experimental settings.  We then compare the experimental and analytical results.

\subsection{Requirements and the Dartmouth Traces}
There are serval publicly available traces online, including the Dartmouth traces \cite{Trace-Dartmouth-Data01}\cite{Trace-Dartmouth-Data02}\cite{Trace-Dartmouth-Data03}, the UCSD traces \cite{Trace-UCSD2}, the IBM-Watson traces \cite{Trace-IBMWatson}, and the Montreal traces \cite{Trace-Montreal}. To choose proper traces, we need to consider the following requirements.  \emph{First}, there should be a large amount of sample points to facilitate an estimation of the user distribution by relative frequency, which is to be compared with the distribution derived by the proposed analysis. Note that the support of the user distribution increases exponentially with the number of cells.  Most available traces do not have a large enough data set. \emph{Second}, the location of cells should be close enough so that there is enough handoff traffic among them to create strong dependency between channel holding times.  Data from already independently operated cells can be analyzed using exiting techniques and hence are not challenging enough to test our analytical model.
To the best of our knowledge, the Dartmouth traces are the most recent public traces satisfying both requirements.  They have been widely studied in the literature \cite{Trace-Dartmouth-Paper01}\cite{Trace-Dartmouth-Paper}\cite{Experiment2}\cite{Mobility-Core1}.
We use data from the academic area in the Dartmouth traces \cite{Trace-Dartmouth-Data03}, a comprehensive record of network activities in a large wireless
LAN (using 802.11b) in Dartmouth College.  The traces includes the data of 152 APs and more than 5000 users, during a 17-week period (Nov.~1, 2003 to Feb.~28, 2004). Most users are students walking on campus.
We focus on the  Simple Network Management Protocol (SNMP) logs of the traces, which are constructed every five minutes, when each AP polls all the users attached to it. Each polling message includes the information such as the name of AP, timestamp, the MAC and IP addresses of users attached to it, signal strength, and the number of packets transmitted.  By analyzing such data, we can derive the average arrival rate, average channel holding time, and the user distribution by relative frequency.

\subsection{Data Preprocessing} \label{sec:preproc}
\subsubsection{Data Extraction}
Since the behavior of users may change greatly between daytime and nighttime, or workdays and holidays, we focus on data accumulated from 9 am to 5 pm on Monday to Friday. We also discard the data accumulated during the periods of holiday breaks, including Thanksgiving (Nov. 26, 2003 to Nov. 30, 2003) and Christmas and New Year (Dec. 17, 2003 to Jan. 4, 2004).
In addition, for some APs, we observed periods when they are temporally power off.  If the total service time of an AP on a certain day is less than 1/3 of its average value, we discard the data for this day.


\subsubsection{Trace Gap Padding}
The session duration is defined as the period of time during which a user is continuously connected to the network. The user may move from one AP to another during a session. Occasionally, a user may disappear from the SNMP report and soon reappear. This may be caused by the user departing and then returning to the network, or due to the missing of an SNMP report. Following the solution proposed in \cite{Mobility-Core1}, we set a departure length threshold $T_d=10$ minutes. Only if a user disappears and reappears within $T_d$, it is regarded as staying in the network and the missing SNMP logs are padded.

\subsubsection{Multiple Association and Ping-Pong Effect}
We also observe that some users are simultaneously associated with multiple APs within a small time interval.  Some even ping-pong among multiple APs.  We use two methods to offset these effects. First, when multiple associations occur, we check the number of packets exchanged with the user. We deem the user is associated with the AP which has exchanged the largest number of packets with the user during its multiple association period.  In addition, if a user leaves one AP and then returns within $5$ minutes, it is regarded as having stayed in the AP.

\subsubsection{Open Users}
a fraction of the users may stay in the system during almost all working hours. These users are regarded as closed users.  Since our analytical model assumes an open network, the closed users are excluded in our experiment.  If a user stays for greater than or equal to $7.5$ hours during working hours on a valid day, it is regarded as a closed user. In our experiment, we observe that $9.91\%$ of all users are closed users.  An analytical model for accommodating closed users is provided in \cite{Mobility-Core1}, which can also be applied to our work.

\subsection{Trace Analysis}

\subsubsection{Poisson Arrivals} \label{section_arrival}

Analysis of the Dartmouth trace in \cite{Mobility-Core1} has shown that the overall \emph{new} session arrivals into the network are well modeled by a Poisson process.  In this work, we further test the arrival process of new sessions \emph{at each AP} against the Poisson assumption.  This is divided into two steps. In the first step, we run an \emph{independence test}, which indicates whether the numbers of arrivals in different time intervals are independent. Since it is not practical to account for all time intervals, we test the independence of arrivals in two consecutive hours at each AP. If the AP passes the test, we regard the arrivals at this AP to be sufficiently independent. Let $H_2$ denote the entropy in the number of new arrivals in two consecutive hours and $H_1$ denote the entropy in the number of arrivals in one hour. Let $\eta=\frac{2H_1-H_2}{H_2}$ be the normalized entropy gap. If $\eta<0.15$, we regard the AP as passing the independence test.  We observe that $144$ of the $152$ APs pass the independence test.

In the second step, we run a \emph{Poisson distribution test}, which indicates whether the number of arrivals is Poisson distributed in a fixed time interval. For each AP that passes the independence test, we count the number of new arrivals in each hour and calculate its real distribution.  Furthermore, by using the actual average arrival rate per hour, we can determine the corresponding theoretical  Poisson distribution. Then, we compute the Kullback-Leibler (KL) divergence $H_{0}$ between the real distribution and the theoretical distribution\footnote{Kullback-Leibler (KL) divergence is a standard approach to measure the difference between two probability distributions $X$ and $Y$. It can be regarded as a measure of the information lost (in bits) when $Y$ is used to represent $X$. When KL divergence is $0$, $Y$ is exactly the same with $X$. {If the KL divergence is small compared with the  entropy  of the distribution $X$, the distribution $Y$ is a close approximation of that of $X$.}}. Let $\theta=\frac{H_{0}}{H_1}$ be the normalized KL value. If $\theta<0.15$, we regard the AP as passing the Poisson distribution test.  We observe that $124$ of the $144$ APs pass the Poisson distribution test.

Those $124$ APs are referred to as \emph{valid APs}, as the new arrivals at these APs can be well approximated as Poisson.  The other $28$ APs are referred to as \emph{invalid APs}.  In our experiments, we study the effects of both including and excluding the non-Poisson new sessions.  We emphasize that the Poisson test is for new arrivals only.  Even for those APs that pass the Poisson test, the overall session arrival process includes both new arrivals and handoff arrivals and hence is non-Poisson.

From the SNMP logs, we observe that the invalid APs tend to have occasional bursty arrivals. Since they are within the academic area, we conjecture that they correspond to large classrooms, which experience periodic rushes a the beginning of lecture hours.  Even though such APs do not match our analytical model, their user distribution is likely easy to predict in practice.

\subsubsection{Number of Stages and Channel Holding Times}

\begin{table}[tbp]
\centering
\renewcommand{\arraystretch}{1.1}
\caption{Number of stages}
\label{table2}
\small
\begin{tabular}{|c|c|c|c|c|c| }
\hline
\bfseries Stages & $1$  &  $2$  & $3$ & $4$ & $\geq 5$   \\
\hline
\bfseries Observations    & $80448$  &  $15767$  & $7410$ &   $3553$ & $6107$    \\
\hline
\end{tabular}
\end{table}

We have collected the distributions of number of stages in each route, which is shown in Table \ref{table2}.  It can be seen that there is a large percentage of sessions staying for just one stage.  To rigorously test the analytical stationary user distribution, we will later present different cases where one-stage sessions are either included or excluded.

Note that if the channel holding times are independently exponentially distributed, our conclusions on the stationary user distribution trivially holds. Therefore, more challenging channel holding times (i.e., arbitrarily distributed and correlated) are necessary to test our analytical results. Fig.~\ref{servicetime} shows the real distributions of channel holding times in different stages. This figure illustrates that none of them are exponentially distributed.  Furthermore, we check the dependency of channel holding times in different stages. The entropies of the distributions of channel holding times at stages $1, 2, 3$ and $4$ are $4.0657$, $3.4172$, $3.3942$ and $2.9792$, respectively, in bits. The entropy of their joint distribution is $10.2998$ bits. Hence, the entropy gap is $4.0657+3.4172+3.3942+2.9792-10.2998=3.5565$ bits, much larger than $0$.  This shows that the channel holding times at different stages are dependent.

\subsubsection{AP Locations and Distance Constraint}
APs that are far away are likely to have little effect on each other, regardless of the mobility and session patterns.  Therefore, to rigorously test the joint distribution of multiple APs, we are more interested in selecting adjacent APs with spatial correlation.  We set a \emph{distance constraint}, under which APs are located pairwisely less than $500$ meters from each other. In the experiments, when we study the joint distribution over multiple APs, this distance constraint is enforced by default, unless otherwise stated. However, we will also present comparison results for
cases with and without it.

\begin{figure}[tbp]
\centering  \hspace{0pt}
\includegraphics[scale=0.47]{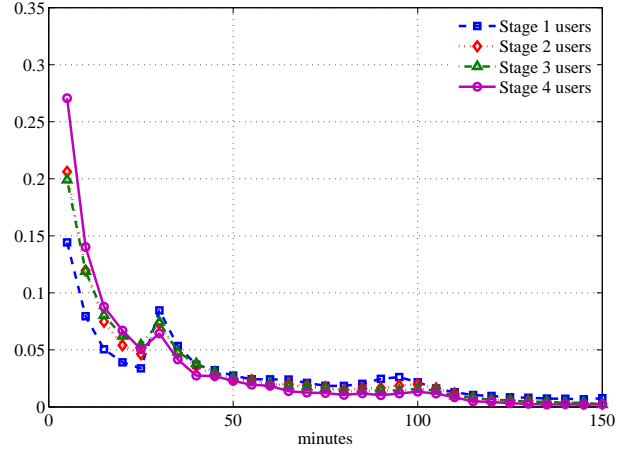}
\caption{The pdf of channel holding time in different stages.}
\label{servicetime}
\end{figure}

\subsection{Marginal User Distribution at a Single AP}

\begin{figure}[t]
\centering  \vspace*{0pt}
\includegraphics[scale=1.15]{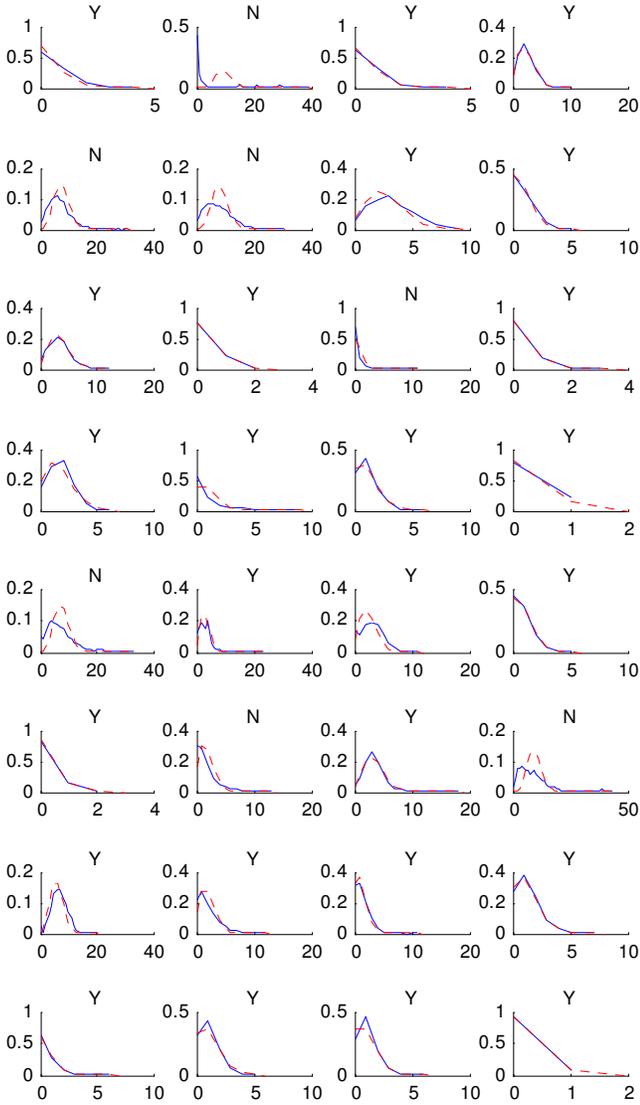}
\caption{Comparison of distributions for single APs. Real distributions  are in solid lines; analytical distributions are in dashed lines.}
\label{SD1}
\end{figure}

We first show the marginal user distribution at individual APs.   For this test, we applied all data after the pre-processing described in Section \ref{sec:preproc}, without further exclusions.
We show a sampling of the $152$ APs. In order to avoid selection bias, we choose APs according to their numeric identity. For each building (with at least one AP), we select the AP with the smallest identity number (i.e., \emph{AP1} if it exists; otherwise, we select \emph{AP2} if it exists; and so forth). There are $32$ buildings with at least one AP, and thus $32$ APs are selected accordingly.

Fig.~\ref{SD1} shows a comparison between the real distributions and the analytical distributions of these APs.  Each subplot is labeled with \emph{Y} or \emph{N}, where \emph{Y} indicates that the AP passes the two-step Poisson test and \emph{N} indicates the opposite.
The figure illustrates that the real distributions and the analytical distributions agree well with each other for those APs that pass the Poisson test.

\subsection{KL Divergence and Entropy Gap for Multiple APs}

 In this paper,   we use KL  divergence $H_{kl}$ to compare the real and analytical joint distributions of multiple APs.
We also test the independence of the numbers of users in different cells by computing
the entropy gap $H_{gap}$, between  the sum of the entropies of real marginal distributions and the entropy of the real joint distribution. The entropy of the real joint distribution $H_{real}$ is also presented for reference. {Note that if $H_{kl}$ is much smaller than $H_{real}$, the analytical distribution  is a close approximation of the real distribution; if $H_{gap}$ is much smaller than $H_{real}$, the numbers of users of single APs are approximately independent.}

Given $n$, the number of APs we aim to study, we randomly choose $n$ different APs. Then we compute $H_{kl}$, $H_{gap}$, and $H_{real}$ with respect to these APs.  By running this procedure $100$ times, we obtain the sample mean and sample standard deviation of $H_{kl}$, $H_{gap}$, and $H_{real}$.   In subsequent studies, we plot the sample mean versus $n$, along with bars showing one sample standard deviation, in Fig. \ref{result1}-\ref{result4}.  Note that the plot points are slightly shifted to avoid overlaps. Because the sample space of user distribution increases exponentially with the number of cells, and the real user distribution is counted through its relative frequency, we limit $n\leq5$ in the experiment in order to ensure enough data are counted for each sample point.

\begin{figure}[tbp]
\centering
\includegraphics[scale=0.42]{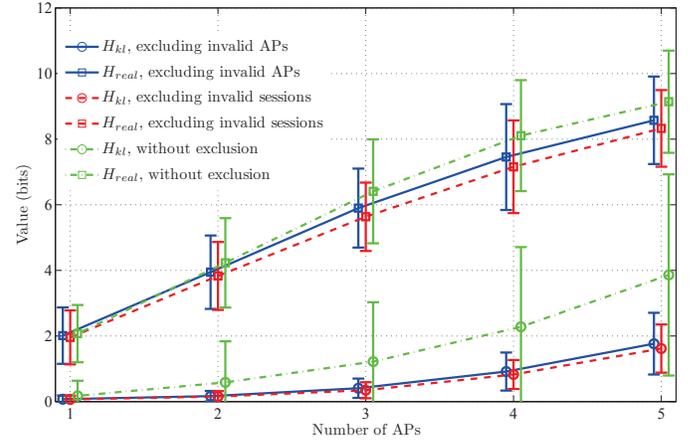}
\caption{$H_{kl}$ and $H_{real}$ under the influence of non-Poisson arrivals.}
\label{result1}
\end{figure}

\begin{figure}[tbp]
\centering
\includegraphics[scale=0.42]{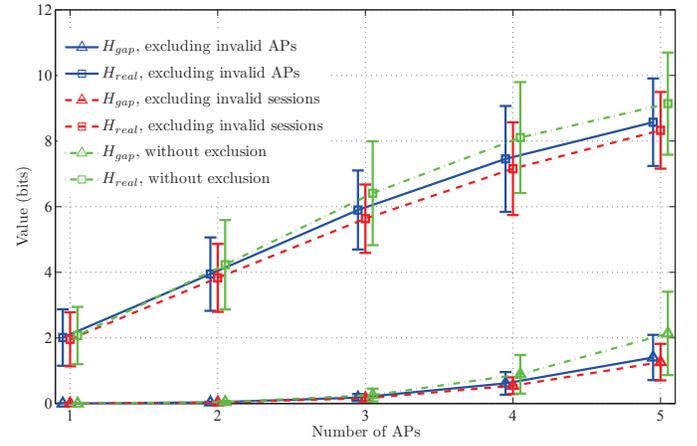}
\caption{$H_{gap}$ and $H_{real}$ under the influence of non-Poisson arrivals.}
\label{result2}
\end{figure}

\subsubsection{Influence of Non-Poisson Arrivals} \label{sec:EffectNon-Poisson}

Clearly, excluding non-Poisson arrivals could improve the accuracy of the analytical model.
We compare $H_{kl}$, $H_{gap}$, and $H_{real}$ under the conditions of either including or excluding non-Poisson arrivals.

A direct method to exclude non-Poisson session arrivals is to remove from the data set all sessions that are initiated at invalid APs.  However, this will reduce the number of handoff session arrivals even in valid APs, hence biasing the analysis.  An alternate approach is to simply remove the invalid APs from the data set, while allowing those non-Poisson sessions to be counted in the valid APs that they pass through. In this way, accurate average arrival rates at the valid APs are maintained.

Thus, we study the following three cases: 1) Excluding sessions initiating at invalid APs (i.e., invalid sessions); 2) Excluding invalid APs; and  3) Without exclusion.  Fig.~\ref{result1} illustrates $H_{kl}$ compared with $H_{real}$ for the three cases, and
Fig.~\ref{result2} illustrates $H_{gap}$ compared with $H_{real}$  for the three cases.
{We observe that both $H_{kl}$ and $H_{gap}$ are much smaller than $H_{real}$, when we either exclude invalid sessions or exclude invalid APs, illustrating that the real distributions are close to the analytical distributions,} and the numbers of users of single APs are approximately independent.
When we do not exclude invalid sessions or invalid APs, $H_{kl}$ and $H_{gap}$ become larger, showing that the analytical distribution is influenced by the non-Poisson arrivals. {However, $H_{kl}$ and $H_{gap}$ remain much smaller than $H_{real}$, illustrating that the analytical distribution is still valid to approximate the real distribution,}  even the arrivals are not strictly Poisson.

In addition, excluding invalid sessions only brings small decrements in $H_{kl}$ and $H_{gap}$ compared with excluding invalid APs. Note that when we exclude invalid sessions, both the one-stage and multiple-stage non-Poisson arrival sessions are excluded; when we exclude invalid APs, only the one-stage non-Poisson arrival sessions are excluded. This illustrates that multiple-stage non-Poisson arrival sessions have only weak influence on the modeling accuracy.


\begin{figure}[tbp]
\centering
\includegraphics[scale=0.42]{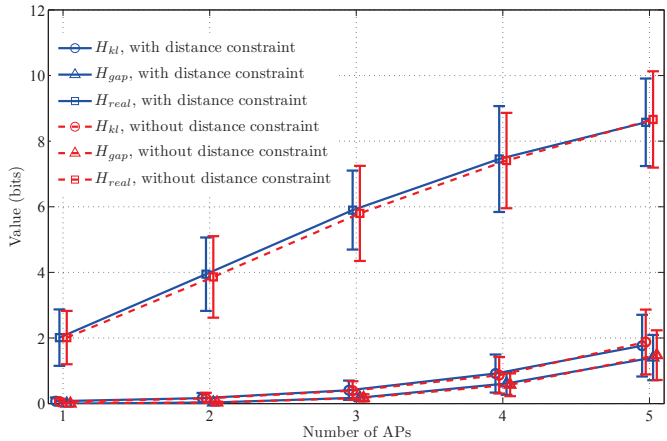}
\caption{$H_{kl}$, $H_{gap}$ and $H_{real}$ under the influence of distance restriction.}
\label{result3}
\end{figure}

\subsubsection{Influence of Distance Constraint}
Fig.~\ref{result3} shows $H_{kl}$, $H_{gap}$, and $H_{real}$ with and without the distance constraint. For both cases, we exclude the invalid APs.  We observe that $H_{kl}$, $H_{gap}$, and $H_{real}$ are nearly unchanged with or without the distance constraint, confirming our expectation that the distance constraint does not influence the accuracy of the analytical model, since the analytical model predicts that the numbers of users of adjacent APs are independent.

\begin{figure}[tbp]
\centering
\includegraphics[scale=0.42]{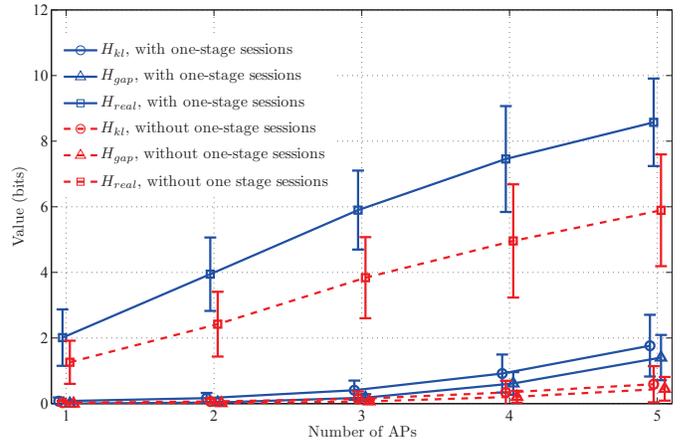}
\caption{$H_{kl}$, $H_{gap}$ and $H_{real}$ under the influence of one-stage sessions.}
\label{result4}
\end{figure}

\subsubsection{Influence of One-Stage Sessions}
Fig.~\ref{result4} shows $H_{kl}$, $H_{gap}$, and $H_{real}$ with and without the one-stage sessions. For both cases, we exclude the invalid APs.
We observe that when we exclude the one-stage sessions, $H_{kl}$ and $H_{gap}$ becomes smaller, suggesting that our model is even more accurate in this case.  This is an apparently counter-intuitive result, since the analytical distribution trivially holds for one-stage sessions.  An explanation for this is the following. Since one-stage sessions are more likely to be new sessions corresponding to attending lectures in a classroom, they are more likely to be non-Poisson.  Since not all non-Poisson arrivals can be excluded by removing the invalid APs, when we further exclude one-stage sessions, we obtain more accurate analytical results.

Note that one-stage sessions can be analyzed as a single-queue model \cite{Mobility-Core4}. Thus, in practice, one may separately analyze one-stage and multiple-stage sessions and combine the resultant user distributions.

%% file: section7-conclusion.tex
\section{Conclusion and Discussion}\label{section_conclusion}
In this paper, we have studied the user distribution in multicell network by establishing a precise analytical model, considering arbitrary user movement and arbitrarily and dependently distributed channel holding times. We have derived the stationary distribution of the number of users in each cell, which is only related to the average arrival rate and the average channel holding time of each cell, and hance is insensitivity with respect to the general movement and session patterns. We have used the Dartmouth trace to validate our analysis, which shows that the analytical model is accurate when new session arrivals are Poisson and remains useful when non-Poisson session arrivals are also included in the data set.

The demonstration of modeling accuracy using an open $M/M/\infty$ Jackson network implies that the number of users in each cell is independently Poisson.  This spatial non-homogeneous Poisson model is commonly used in the geometric analysis of interference in wireless networks \cite{Book-StochasticGeometry1, Book-StochasticGeometry2}.  It can alternatively be inferred from associating the user trajectory as location-dependent marks to a space-time Poisson process representing the entry location and time of the users \cite{Book-StochasticNetwork}.  Our results additionally show that the mean values for the Poisson distributions in different cells are insensitive to the arbitrary and dependent channel holding times. This enables simple yet accurate computation of related performance measures in a complex mobile system.